\documentclass[12pt]{article}
\usepackage{amsbsy,amsmath,amsthm,amsfonts,amssymb,times,subfigure,graphics,graphicx,bm,mathtools,commath}
\usepackage{multirow}
\usepackage{multicol}
\usepackage{thmtools}
\usepackage{siunitx}
\usepackage{lipsum,xcolor}
\RequirePackage{natbib}
\usepackage{algorithm,algorithmic}
\usepackage{mathrsfs}
\usepackage{caption}
\usepackage{wrapfig}

\declaretheorem[style=definition]{example}
\newtheorem{theorem}{Theorem}
\newtheorem{lemma}{Lemma}
\newtheorem{corollary}{Corollary}
\newtheorem{proposition}{Proposition}
\theoremstyle{definition}

\newtheorem{remark}{Remark}

\renewcommand\thmcontinues[1]{Continued}

\newcommand{\sX}{\textsf{X}}

\newcommand{\pcite}[1]{\citeauthor{#1}'s \citeyearpar{#1}}
 
\usepackage [english]{babel}
\usepackage [autostyle, english = american]{csquotes}
\MakeOuterQuote{"}

\begin{document}

\title{A geometric approach to informed MCMC sampling}
%
\author{Vivekananda Roy\\ Department of Statistics, Iowa State University, USA}
 \date{}
\maketitle              

\begin{abstract}
  A Riemannian geometric framework for Markov chain Monte Carlo (MCMC)
  is developed where using the Fisher-Rao metric on the manifold of
  probability density functions (pdfs), informed proposal densities
  for Metropolis-Hastings (MH) algorithms are constructed. We exploit
  the square-root representation of pdfs under which the Fisher-Rao
  metric boils down to the standard $L^2$ metric on the positive
  orthant of the unit hypersphere. The square-root representation
  allows us to easily compute the geodesic distance between densities,
  resulting in a straightforward implementation of the proposed
  geometric MCMC methodology. Unlike the random walk MH that blindly
  proposes a candidate state using no information about the target,
  the geometric MH algorithms move an uninformed base density (e.g., a
  random walk proposal density) towards different global/local
  approximations of the target density, allowing effective exploration
  of the distribution simultaneously at different granular levels of
  the state space. We compare the proposed geometric MH algorithm with
  other MCMC algorithms for various Markov chain orderings, namely the
  covariance, efficiency, Peskun, and spectral gap orderings. The
  superior performance of the geometric algorithms over other MH
  algorithms like the random walk Metropolis, independent MH, and
  variants of Metropolis adjusted Langevin algorithms is demonstrated
  in the context of various multimodal, nonlinear, and high
  dimensional examples. In particular, we use extensive simulation and
  real data applications to compare these algorithms for analyzing
  mixture models, logistic regression models, spatial generalized
  linear mixed models and ultra-high dimensional Bayesian variable
  selection models. A publicly available R package accompanies the
  article.
\end{abstract}


  \section{Introduction}
  \label{sec:int}

  Sampling from complex, high dimensional, discreet, and continuous
  probability distributions is a common task arising in diverse
  scientific areas, such as machine learning, physics, and
  statistics. Markov chain Monte Carlo (MCMC) is the most popular
  method for sampling from such distributions and among the
  different MCMC algorithms, Metropolis-Hastings (MH) algorithms
  \citep{metr:rose:rose:rose:tell:tell:1953, hast:1970} are
  predominant. In MH algorithms, given the current state $x$, a
  proposal $y$ is drawn from a density $f(y|x)$, which is then
  accepted with a certain probability. The random walk MH (RWM)
  algorithms use a symmetric $f$, that is $f(y|x) = f(x|y)$, for
  example when $y= x+ I$ with the increment $I$ following a
  normal/uniform density centered at the origin
  \citep[][chap. 7.5]{robe:case:2004}. The RWM algorithms are easy to
  implement, but since the proposal density $f$ does not use any
  information on the target density $\psi$, RWM can suffer from slow
  convergence, particularly in high dimensions \citep{neal:2003}. This
  led to the development of alternative MH proposals that exploit some
  information about the target density $\psi$. Intuitively, the informed MCMC
  schemes can avoid frequent visits to states with low target probabilities
  leading to faster convergence. Indeed, for sampling from continuous
  target densities, informative MH proposals such as those of the
  Metropolis adjusted Langevin algorithms (MALA)
  \citep{ross:doll:frie:1978, besa:1994, robe:twee:1996b} and
  Hamiltonian Monte Carlo (HMC) algorithms
  \citep{duan:kenn:pend:rowe:1987} have been constructed employing
  the gradient of the log of the target density.

  However, the standard MALA and HMC algorithms do not efficiently
  sample from high dimensional distributions with complex structures
  such as strong dependencies between variables, non-Gaussian shapes,
  or multiple modes \citep{giro:cald:2011, beta:2013,
    roy:zhan:2023}. Also, as mentioned in \cite{giro:cald:2011}, `the
  tuning of these MCMC methods remains a major issue'. The Euclidean
  MALA and HMC algorithms fail to take into account the geometry of
  the target distribution in the selection of step sizes. Indeed,
  using ideas from information geometry, \cite{giro:cald:2011}
  constructed manifold MALA and HMC methods called the manifold MALA
  (MMALA) and the Riemannian manifold HMC (RMHMC), respectively.
  MMALA and RMHMC adapt to the second-order geometric structure of the
  target, which allows these algorithms to align their proposals in
  the direction of the target that exhibits the greatest local
  variation and generally outperform their Euclidean counterparts in
  exploring high dimensional complex target distributions
  \citep{brof:roy:2023, giro:cald:2011}. However, the sophisticated
  form of the Hamiltonian employed in RMHMC necessitates the use of
  complex numerical integrators that are significantly more expensive
  than the numerical integrator employed in HMC. Since these manifold
  variants of MALA and HMC use a position-specific metric, it needs to
  be recomputed in every iteration, and generally, this computation
  scales cubically.  Also, for implementing these manifold chains,
  first and higher-order derivatives of the log target density are
  required. One must find appropriate alternatives if these
  derivatives are not available in closed form.  Also, MALA and HMC
  algorithms are not applicable for discrete distributions, although
  there have been some recent developments to extend these methods to
  form informative MH proposals for discrete spaces by emulating the
  behavior of gradient-based MCMC samplers on Euclidean spaces
  \cite[see
  e.g.][]{zane:2020,zhan:liu:2022,nish:duns:2020,pakm:pani:2013}. On
  the other hand, as mentioned in \cite{zane:2020}, it is `typically
  not feasible' to sample from their informed proposals in continuous
  state spaces.

 In this article, we propose an original Riemannian geometric framework for developing
 informative MH proposals, irrespective of whether the state space is
 discrete or continuous. In RWM, the current state $x$ is blindly
 moved to $x+I$ following a Euclidean random perturbation $I$, failing
 to take into account the non-Euclidean nature of the target space. On
 the other hand, starting with a `base', uninformed kernel $f$, our
 proposed method moves $f$ in the directions of $\psi$ to produce
 informed MH proposals. Our novel formulation has the merit of being
 straightforward and universally applicable to both discrete and
 continuous spaces exploiting their natural Riemannian geometry. The
 Fisher-Rao (FR) metric that we consider here is the `natural'
 metric on the space of probability density functions (pdfs)
 \citep{rao:1945} and it is known that the gradient under the FR
 metric is the fastest ascending direction of a distribution objective
 function \citep{amar:1998}. Thus, the FR metric should ideally be used
 to explore a distribution, although, as noted before, for manifold
 variants of MALA and HMC, the use of this metric generally leads to higher computational
 burden.

 To build computationally efficient, informative MCMC algorithms that
 adapt to the geometry of the target, we consider a novel approach
 using \pcite{bhat:1943} `square-root' representation for pdfs. Under
 this representation, the manifold of pdfs can be identified with the
 positive orthant of the unit sphere and the FR metric boils down to
 the standard $L^2$ metric (see Figure~\ref{fig:sqrttrans}). This
 simplifies computations through the availability of explicit,
 closed-form expressions for useful geometric quantities like geodesic
 paths, distances as well as exponential and inverse-exponential maps.
 Thus, the proposed general-purpose, geometric MCMC algorithms, unlike
 MMALA and RMHMC, do not require the first and higher-order
 derivatives of the log target density. Also, \cite{zane:2020}
 considers point-wise informed proposals for discrete spaces, whereas
 here, we construct informed proposals that incorporate both `local'
 moves of the base kernel and `global' moves respecting the geometry
 of the space of pdfs applicable to both discrete and continuous
 spaces. The concurrent global-local steps enable moves between the
 modes without augmenting the state space with a `temperature'
 variable as in a tempering scheme \citep{geye:1991,mari:pari:1992}
 and simultaneous exploration of the modal regions.
 
We provide new results based on Peskun, covariance,
efficiency and spectral gap orderings for comparing Markov chains with the same stationary distribution. These general theoretical
results are then used to demonstrate the improvement obtained by the
proposed geometric method over the uninformed base Markov chains. It
is known that the spectral gap is closely related to the convergence
properties of a Markov chain. Thus, the geometric MCMC leads to
superior convergence properties over the base RWM, independent MH, or other Markov chain kernels. For RWM algorithms to be geometrically ergodic, it is
necessary that the invariant density $\psi$ has moment generating
function \citep{jarn:twee:2003}, although for heavy-tailed target
distributions, the RWM chains can have a polynomial rate of convergence
\citep{jarn:robe:2007}. \cite{meng:twee:1996} proved that when the
support of $\psi$ is $\mathbb{R}$, the RWM chain cannot be uniformly
ergodic. The geometric, polynomial, and uniform
ergodicity definitions can be found in \cite{douc:moul:2018}. We provide examples
where the proposed geometric MH chain is uniformly ergodic, whereas
the MH chains with the base RWM or independent kernels are not even
geometrically ergodic. \cite{john:geye:2012} considered RWM for
densities induced by appropriate transformations to obtain geometric
ergodicity for RWM algorithms even when the original target density is
sub-exponential. This variable transformation method works for
densities with continuous variables, and as mentioned in
\cite{john:geye:2012}, it may cause other problems. For example, the
induced density can be multimodal even if the original density is not.
This article's extensive examples involving discrete and
continuous spaces using simulated and real data show
orders of magnitude improvements in geometric MCMC compared to RWM,
independent MH, and other MCMC algorithms.

The rest of the article is organized as
follows. Section~\ref{sec:sqrt} introduces the square-root
representation and the FR Riemannian geometric framework. In
Section~\ref{sec:mmh}, we define the informed proposal distributions
constructed by `moving' any base uninformed kernel in the directions
of the target density and its local and global approximations.
Section~\ref{sec:mcorder} provides some results for comparing general
state space Markov chains. Then, Section~\ref{sec:indgeo} uses these
general results to establish the superior performance of the proposed
geometric MCMC algorithms over the MH algorithms based on the
uninformed base kernel. Section~\ref{sec:samph} describes some methods
for sampling from the geometric proposal distributions for efficient
simulation using the proposed MH algorithms. In
Section~\ref{sec:exam}, we consider several widely used high
dimensional and nonlinear models with complex, multimodal target
distributions, namely mixture models, logistic models, spatial
generalized linear mixed models (GLMMs) and the Bayesian variable
selection models with spike and slab priors. While the target
distribution for the feature selection example is discrete, it is
continuous for the other examples. The proposed geometric methods
outperform several traditional and state-of-the-art MCMC algorithms in
all these examples. Finally, in Section~\ref{sec:disc} we discuss
possible extensions and future works. The supplementary material
contains an alternative formulation of the geometric MCMC method,
proofs of the theoretical results, additional technical derivations
for some MCMC algorithms, and extra plots and numerical results from
additional real data analyses and simulation studies for the various
examples considered in the paper.

\section{Square-root representation and Fisher-Rao metric}
\label{sec:sqrt}

Let
$\mathcal{F} = \{f: f: \mathbb{R} \rightarrow \mathbb{R}_+ \cup \{0\}
\; \mbox{with}\; \int_{\mathbb{R}} f(x) dx =1\}$ be the space of all
pdfs on $\mathbb{R}$. To keep the
notations simpler, in this section, we make the presentation in the context of
densities on $\mathbb{R}$ although it straightforwardly extends to
higher dimensions. For
$f \in \mathcal{F}$,
$T_{f} \mathcal{F} = \{ \partial f : \mathbb{R} \rightarrow \mathbb{R}
\; \mbox{with}\; \int_{\mathbb{R}} \partial f(x) f(x) dx = 0\}$ is the
tangent space, which is a linear space. Thus, $T_{f} \mathcal{F}$ can
be viewed as the set of all possible perturbations of $f$. The
Fisher-Rao (FR) metric \citep{rao:1945} is defined as
\[
  \langle \partial f_1,  \partial f_2 \rangle_f = \int_{\mathbb{R}} \partial f_1(x) \partial f_2(x) \frac{1}{f(x)} dx.
  \]
  Although the FR metric has several advantages, computing geodesic
  paths and distances for it are difficult. One solution is to
  consider \pcite{bhat:1943} `square-root' representation
  $\xi: \mathcal{F} \rightarrow \mathcal{M}$ given by $\xi (f) = \sqrt{f}$,
  where
  $\mathcal{M} = \{\rho: \mathbb{R} \rightarrow \mathbb{R}_+\cup \{0\}
  \; \mbox{with}\; \int_{\mathbb{R}} \rho^2(x) dx = 1\}$ (see Figure~\ref{fig:sqrttrans}). 
  We will
  use the inverse map $\xi^{-1} (\rho) = \rho^2 = f$ to take $\rho$
  back to $\mathcal{F}$, the space of densities. The square-root
  representation has previously been used in shape analysis
  \citep{sriv:klas:2010}, variational Bayes
  \citep{saha:bhar:kurt:2019}, sensitivity analysis
  \citep{kurt:bhar:2015}, quantum estimation theory
  \citep{facc:kim:2016} among other areas.
  \begin{figure}%
    \centering
    \includegraphics[width=0.5\textwidth]{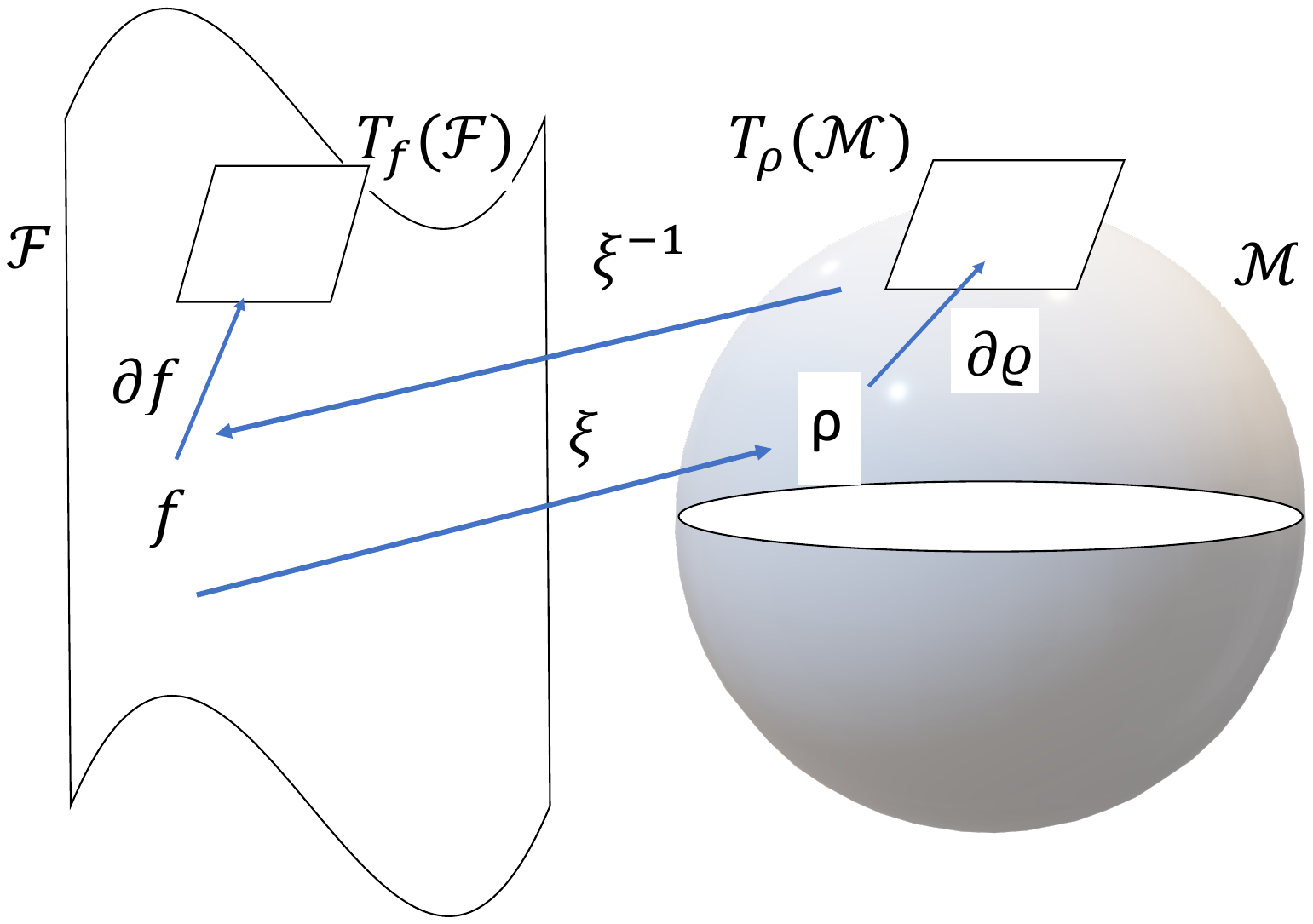}
    \caption{The square root transformation from $\mathcal{F}$, the space of pdfs to
      $\mathcal{M}$, the positive orthant of the unit sphere.}%
    \label{fig:sqrttrans}%
\end{figure}

  Note that $\mathcal{M}$ is the positive orthant of the unit sphere,
  where the FR metric boils down to the standard $L^2$ metric
  $\langle \rho_1, \rho_2 \rangle = \int_{\mathbb{R}} \rho_1(x)
  \rho_2(x) dx$ and geodesic paths and distances are available in
  closed form. Indeed, the geodesic distance between $\rho_1$ and
  $\rho_2$ in $\mathcal{M}$ is the angle between them
  $\theta = \cos^{-1} \langle \rho_1, \rho_2 \rangle$. Note that, for
  identical distributions $\rho_1 = \rho_2$ implying $\theta
  =0$. Also, by Jensen's inequality,
  $\langle \rho_1, \rho_2 \rangle = \langle f_1, \rho_2/\rho_1 \rangle
  \le 1$ implying that $\theta$ is bounded above by the right angle
  providing an upper bound to the geodesic distance between
  pdfs. Also, the geodesic path between $\rho_1$ and $\rho_2$ indexed
  by $r \in [0, 1]$ is
  $\varsigma(r) = [\sin (\theta)]^{-1} [\rho_1 \sin (\theta - r \theta) +
  \rho_2 \sin (r \theta)]$.

  Note that, for $\rho \in \mathcal{M}$, the tangent space
  $T_\rho \mathcal{M} = \{\partial \rho : \langle \partial \rho, \rho
  \rangle =0\}$ is a linear space. In analogy to vector addition
  $x + v$ in linear space that moves a point $x$ along the straight
  line in the direction of $v$, we define the exponential map
  $\exp_{\rho_1} : T_{\rho_1} \mathcal{M} \rightarrow \mathcal{M}$ as
  moving a point $\rho_1$ along the geodesic tangent to $\partial \rho$
  at $\rho_1$ and is defined as
  $\exp_{\rho_1} (\partial \rho) = \cos(\|\partial \rho\|) \rho_1 +
  \sin (\|\partial \rho\|) \partial\rho (\|\partial\rho\|)^{-1}$ where
  $\|\cdot \|$ is the $L_2$ norm.  The inverse exponential
  map
  $\exp^{-1}_{\rho_1} : \mathcal{M} \rightarrow T_{\rho_1}
  \mathcal{M}$ is given by
  $\exp^{-1}_{\rho_1} (\rho_2) = \theta [\sin (\theta)]^{-1}(\rho_2 -
  \cos (\theta) \rho_1)$ and it is used to map points from the
  representation space to the tangent space.

Unlike a linear space for which the tangent space is the same
  everywhere, it is not true for a general manifold. Fortunately, the
  parallel transport $\Gamma_{\rho_1}^{\rho_2}: T_{\rho_1} \mathcal{M}
  \rightarrow T_{\rho_2} \mathcal{M}$ provides a link between tangent
  spaces at different points along geodesic paths (great circles) in
  $\mathcal{M}$.  For $\partial \rho \in T_{\rho_1} \mathcal{M},
  \Gamma_{\rho_1}^{\rho_2} (\partial \rho) = \partial \rho - 2
  \langle \partial \rho, \rho_2 \rangle (\rho_1 + \rho_2)/ \|\rho_1 +
  \rho_2\|$.

\section{Manifold MH proposals}
\label{sec:mmh}
Let $\psi$ be the target density on ${\sX}$ with respect to some
measure $\mu$. MCMC algorithms simulate a Markov chain $\{X^{(n)}\}_{n
\ge 1}$ with some Markov transition function (Mtf) $Q(x, dy)$ that has
$\psi$ as its stationary density. While appropriate choices of $Q$
that uses information on $\psi$ results in a fast mixing Markov chain,
bad, uninformed choices can take weeks or even months to converge to
$\psi$.  To develop informative MCMC schemes that adapt to the
geometry of the target distribution, starting with any `baseline'
proposal density $f(y|x)$, we construct geometric MH proposals by
perturbing $f$ in the directions of $\psi$ or some approximations of
$\psi$.

 Let $\mathcal{G} = \{g_1, g_2, \dots, g_k\} $ be a set of $k$ pdfs
representing different (local/global) approximations of the target
density. Later in this section, we discuss possible choices for $\mathcal{G}$. For a given `baseline'
density $f$, the inverse exponential map $\chi_{g_i} \equiv
\exp^{-1}_{\sqrt{f}}(\sqrt{g_i})$ takes the density to the tangent
space $T_{\sqrt{f}} \mathcal{M}$, 
and then for a given step size $\epsilon$, the exponential map
$\exp_{\sqrt{f}}(\epsilon \exp^{-1}_{\sqrt{f}}(\sqrt{g_i})) =
\exp_{\sqrt{f}}(\epsilon \chi_{g_i})$ takes it back to $\mathcal{M}$
(Figure~\ref{fig:perturb}). Using the transformation $\xi^{-1}$
defined in Section~\ref{sec:sqrt}, we construct a `perturbed' pdf
$\{\exp_{\sqrt{f}}(\epsilon \chi_{g_i})\}^2$.
For this perturbed density by varying $\epsilon$ from zero and one,
the geodesic path from $f$ to $g_i$ can be traced. Note that,
$\{[\exp_{\sqrt{f}}(\epsilon \chi_{g_i})]^2, i=1,\dots,k\}$ is the set
of $\epsilon$-perturbations of $f$ along the directions specified by
$\mathcal{G}$. We use these geometrically perturbed densities to form
informative proposals for MH algorithms. Indeed, the proposed
geometric proposals use these densities to move an uninformed baseline
density $f$ along different approximations of the target density
$\psi$.

Recall that, in an RWM chain, given the current state $x$, an increment
(step size) $I$ is proposed according to a fixed density to move to
the candidate point $y= x+I$. Thus, the random perturbation $I$ used
for local exploration in the RWM algorithm does not take into account
the non-Euclidean structure of the target parameter space. On the
other hand, the step size and the directions used in the
geometric proposal density respect the geometry of pdfs.

\begin{figure}%
    \centering
    \includegraphics[width=0.5\textwidth]{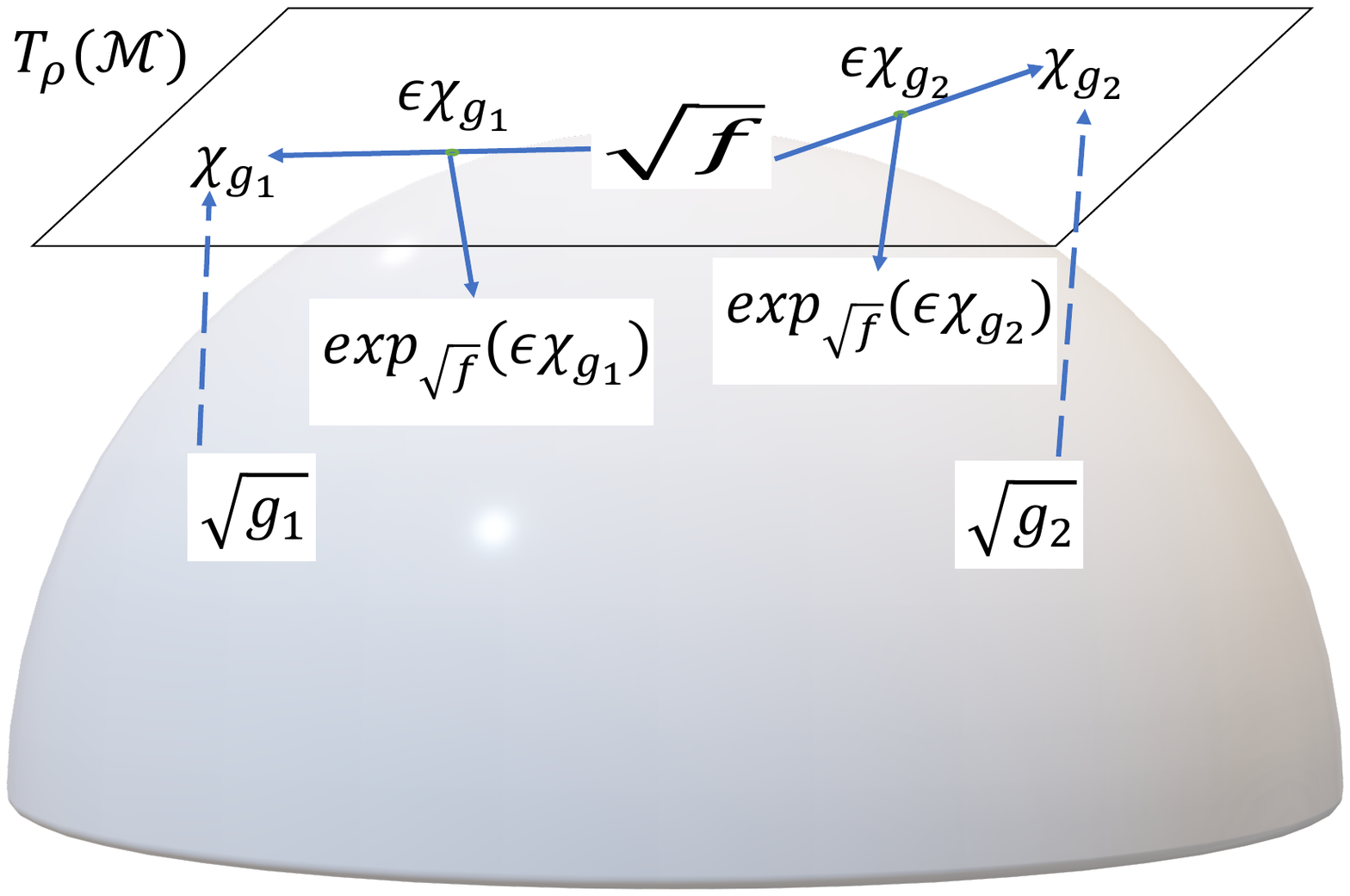}
  \caption{Geometric $\epsilon$ perturbation of $\sqrt{f}$ in the direction of $\sqrt{g_i}, i=1,2$.}%
    \label{fig:perturb}%
\end{figure}

\begin{proposition}
    \label{prop:perturb}%
     Given densities $f$ and $g_i$, $\{\exp_{\sqrt{f}}(\epsilon \exp^{-1}_{\sqrt{f}}(\sqrt{g_i}))\}^2(y|x)$ is a pdf for $\epsilon \in [0, 1]$, with
\begin{align}
      \label{eq:prop}
   &  \big\{\exp_{\sqrt{f}}(\epsilon \exp^{-1}_{\sqrt{f}}(\sqrt{g_i}))\big\}^2(y|x) \nonumber\\ &= \cos^2(\epsilon \theta_{i,x}) f(y|x) + \sin^2(\epsilon \theta_{i,x}) h_i(y|x) + \sin(2\epsilon \theta_{i,x}) \sqrt{f(y|x)} \zeta_i(y|x),
    \end{align}
    where $\theta_{i,x} = \cos^{-1} \big\langle \sqrt{f(y|x)}, \sqrt{g_i(y|x)} \big\rangle$,
    \begin{equation}
  \label{eq:defzeta}
\zeta_i(y|x) = \frac{\sqrt{g_i(y|x)} - \sqrt{f(y|x)} \big\langle
\sqrt{f(y|x)}, \sqrt{g_i(y|x)} \big\rangle}{\sqrt{1 - \big\langle \sqrt{f(y|x)}, \sqrt{g_i(y|x)} \big\rangle^2}},  
\end{equation}
and $h_i(y|x) = \zeta^2_i(y|x)$.
\end{proposition}
    
 Ignoring the last term in \eqref{eq:prop} we consider the
    following geometric proposal density 
    \begin{equation}
      \label{eq:mixprop}
            \phi_{i, \epsilon}(y|x) \equiv \cos^2(\epsilon \theta_{i,x}) f(y|x) + \sin^2(\epsilon \theta_{i,x}) h_i(y|x),
          \end{equation}
          which is a mixture of the densities $f$ and          
\begin{equation}
  \label{eq:defh}
h_i(y|x) = \zeta_i^2 (y|x) = \frac{\big(\sqrt{g_i(y|x)} - \sqrt{f(y|x)} \big\langle
\sqrt{f(y|x)}, \sqrt{g_i(y|x)} \big\rangle\big)^2}{1 - \big\langle \sqrt{f(y|x)}, \sqrt{g_i(y|x)} \big\rangle^2}.  
\end{equation}

Let $a = (a_1,\dots,a_k)$ be a probability vector, that is,
$a_i \ge 0$ with $\sum_{i=1}^k a_i =1$.  We now describe our proposed
geometric MH algorithm. Suppose $X^{(n-1)} = x$ is the current value
of the Markov chain.
\begin{algorithm}[H]
\caption{The $n$th iteration}
\begin{algorithmic}[1]
  \label{alg:mmhmix}
  \STATE Draw $y \sim \phi_{\epsilon} (y|x) \equiv \sum_{i=1}^k a_i \phi_{i,\epsilon} (y|x)$.
  \STATE Set 
    \begin{equation}
      \label{eq:accep}
    \alpha(x, y) = \mbox{min} \Big\{\frac{\psi(y) \phi_{\epsilon}(x|y)}{\psi(x) \phi_{\epsilon}(y|x)}, 1 \Big\}.
    \end{equation}
  \STATE Draw $\delta \sim$ Uniform $(0, 1)$. If $\delta < \alpha(x, y)$
  then set $X^{(n)} \leftarrow y$, else set $X^{(n)} \leftarrow x$.
\end{algorithmic}
\end{algorithm}
In Section~\ref{sec:samph}, we describe an efficient method for
sampling from the geometric MH proposal density
\eqref{eq:mixprop}. 
The step size $\epsilon$ determines the amount of perturbation of the
base density $f$ in the directions of $\mathcal{G}$. By changing
$\epsilon$ from zero to one, the geodesic paths from $f$ to the
densities in $\mathcal{G}$ are traced. Thus, smaller values of
$\epsilon$ lead to higher acceptance rates. So, given the base density
$f$ and the densities $g_i$'s, the single parameter $\epsilon$ provides a simple way of tuning
the geometric MCMC chains based on the acceptance rate criteria. We
now discuss some choices for the baseline pdf $f$ and useful global
and local approximations of $\psi$. 
Given the current state $x$, $f(y|x)$ can be the normal/uniform
density 
centered at the current state $x$, that is, $f$ is the proposal density often used in the RWM
algorithms mentioned in the Introduction. Or, $f$ can be the proposal
density used for an independent MH algorithm
\citep[][chap. 7.4]{robe:case:2004}, in which case $f(y|x) =f(y)$.
Now, we discuss possible choices for $\mathcal{G}$. For multimodal
targets these can be a set of densities centered at the local modes
(see Example~\ref{sec:exampmixture} in Section~\ref{sec:exam}), it can
be the target density constrained on an appropriate neighborhood of
the current state $x$ (see Example~\ref{sec:exampvarsel} in
Section~\ref{sec:exam}), these densities can be some normal
approximations to the target posterior density (see Example~\ref{sec:examplogistic} in
Section~\ref{sec:exam}) or the conditional posterior density (see Example~\ref{sec:exampsglmm} in
Section~\ref{sec:exam}) according to the
Bernstein-von Mises theorem or some variational approximations of the
target posterior density \citep{blei:kucu:2017}. Other possibilities
of $g_j$ can be
\begin{equation}
\label{eq:choiceg}
g_j(y|x) = \prod_{i=1}^d\psi(y_i|x_{-i}), \; \mbox{or} \; g_j(y|x) =
\prod_{i=1}^d\Big[\frac{1}{2}\psi(y_i|x_{-i}) + \frac{1}{2}
\frac{\psi(y_i|x_{-i})^{\tau_j}}{\int_{\sX_i}\psi(y_i|x_{-i})^{\tau_j}
dy_i} \Big],
\end{equation}
for some $\tau_j \in (0, 1)$, where $\psi(y_i|x_{-i})$'s are the full
conditionals, $x$ is the current state of the Markov chain, and
$\sX= \sX_1 \times \cdots \sX_d$. For the first choice in
\eqref{eq:choiceg}, the density $g_j$ does not change with $j$ and
thus $k=1$ in this case. For the second choice, the marginals of
$g_j(y|x)$ are mixtures of $\psi(y_i|x_{-i})$ and a flattened version
of it, which is suggested by \cite{zane:robe:2019} as a robust and
efficient choice for their Tempered Gibbs Sampler.

\begin{remark}
  Instead of considering a single MH proposal based on a mixture of
  $\phi_{i,\epsilon}$ as in Algorithm~\ref{alg:mmhmix}, a mixture of
  MH Mtfs based on each $\phi_{i,\epsilon}$ can also be considered
  (see Algorithm 2 in the supplement).  \cite{tier:1998} showed that
  Algorithm~\ref{alg:mmhmix} dominates the later algorithm in the
  Peskun sense. On the other hand, per iteration computation cost of
  Algorithm~\ref{alg:mmhmix} is higher than that of the later. The
  acceptance probability of Algorithm~\ref{alg:mmhmix} requires
  computation of $\phi_{i,\epsilon}$ for all $i \in \{1,\dots,k\}$,
  whereas that of the mixture of MH algorithms needs computing
  $\phi_{i,\epsilon}$ only for the sampled $i$.
\end{remark}

\begin{remark}
  The selection probability $a_i$ determines the proportion of moves
  in a specific direction. Intuitively, if certain approximation $g_i$
  is believed to be closer to the target, or if $g_i$ corresponds to a
  density around a local mode of $\psi$ with higher mass, then
  larger $a_i$ values can be used. In our empirical experiments, we
  have observed that setting $a_i=1/k, i=1,\dots,k$ is a safe choice
  as it leads to similar performance as the informative values of $a$.
  These probabilities $a_i$'s can also be chosen adaptively or
  according to some distributions \citep{levi:case:2006}.
\end{remark}
  \section{Ordering Markov chains}
  \label{sec:mcorder}
    In this section, we introduce some orderings for general state space
Markov chains, and later, we will use these results to compare geometric
MH chains with other MCMC algorithms. Let $P(x, dy)$ be a Mtf on $\sX$, equipped with a countably
generated $\sigma-$algebra $\mathbb{B}(\sX)$.  Let $\{X^{(n)} \}_{n
\geq 1}$ denote the Markov chain driven by $P$.  Throughout we assume
that $\{X^{(n)} \}_{n \geq 1}$ ($P$) is a {\it Harris ergodic}
chain with stationary density $\psi$. Thus, the estimator $\bar{t}_n := \sum_{i=1}^n t(X^{(i)})/n$ is
strongly consistent for $E_{\psi} t := \int_{\sX} t(x) \psi(x)\mu(dx)$ for all real-valued
functions $t$ with finite mean, no
matter what the initial distribution of $X_1$ is \cite[][chap. 17]{meyn:twee:1993}. We
say a central limit theorem (CLT) for $\bar{t}_n$ exists if for some
positive, finite quantity $v(t, P)$, $\sqrt{n} (\bar{t}_n - E_{\psi}
t) \stackrel{\text{d}}{\longrightarrow} N(0, v(t, P))$, as $n
\rightarrow \infty$.

Let $L^2(\psi)$ be the vector space of all real-valued, measurable
functions on $\sX$ that are square integrable with respect to
$\psi$. The inner product in $L^2(\psi)$ is defined as
$\langle t, s \rangle_\psi = \int_\sX t(x) \, s(x) \, \psi(x) \mu(dx) \,$. For two
Mtf's $P$ and $Q$ with the invariant density $\psi$, $P$ is said
to be more efficient than $Q$ if $v(t, P) \le v(t, Q)$ for all
$t \in L^2 (\psi)$ \citep{mira:geye:1999}. Another method of comparing
Markov chains is the {\it Peskun ordering} due to \cite{pesk:1973}
which was later extended to general state space Markov chains by
\cite{tier:1998}. The Mtf $P$ dominates $Q$ in the Peskun sense,
written $P \succeq Q$ if for $\psi$-almost all $x$,
$P (x, A \setminus \{x\}) \ge Q (x, A \setminus \{x\})$ for all
$A \in \mathbb{B}(\sX)$. In this case, a Markov chain driven by $P$ is
less likely to be held back in the same state for succeeding times
than a Markov chain driven by $Q$.

In order to define some other notions of comparing Markov chains,  note
that the Mtf $P$ defines an operator on $L^2(\psi)$ through,
$(P t) (x) = \int_{\sX} t(y) P(x, dy)$. Abusing notation, we use $P$
to denote both the Mtf and the corresponding operator. The Mtf $P$
is reversible with respect to $\psi$ if for all bounded functions
$t, s \in L^2(\psi)$, $\langle Pt, s \rangle_\psi = \langle t, Ps
\rangle_\psi$. The spectrum of the operator $P$ is defined as
\[
\sigma(P) = \Big\{ \lambda \in \mathbb{R}: P - \lambda I \;\,\mbox{is
  not invertible} \Big\}, \;
\]
where $I$ is the identity operator. For reversible $P$, it follows
that $\sigma(P) \subseteq [-1,1]$. Define the operator $P_0$ as
$P_0 t= Pt - E_{\psi} t$, for $t \in L^2(\psi)$. The speed of
convergence of the Markov chain with Mtf $P$ to stationarity
($\psi$) is determined by its spectral gap,
$\mbox{Gap}(P) := 1- \mbox{sup} \{|\lambda|: \lambda \in
\sigma(P_0)\}$. The Markov chain with Mtf $P$ converges at least as
fast as the Markov chain with Mtf $Q$ if
$\mbox{Gap}(P) \ge \mbox{Gap}(Q)$. Finally, $P$ dominates $Q$ in the
covariance ordering if $\langle Qt, t \rangle_\psi \ge \langle Pt, t \rangle_\psi$ for every
$t \in L^2(\psi)$, that is, the lag one autocorrelations of a
stationary Markov chain driven by $P$ are at most as large as for a
chain with Mtf $Q$. \cite{tier:1998} established that if $P \succeq Q$, then $P$ dominates
$Q$ both in covariance and efficiency
orderings. Theorem~\ref{thm:orde} shows that a modified Peskun
ordering implies modified covariance, efficiency as well as spectral gap orderings.
Let
  $\sigma^2_t = E_{\psi} t^2 -[E_{\psi} t]^2$ and
  $L^2_0(\psi) = \{s \in L^2(\psi): E_{\psi}s =0\}$.
  
\begin{theorem}
  \label{thm:orde}
  Let $P$ and $Q$ be two Mtf's with the same invariant density
  $\psi$. Assume that for $\psi$-almost all $x$ we have
  $P (x, A \setminus \{x\}) \ge c Q (x, A \setminus \{x\})$ for all
  $A \in \mathbb{B}(\sX)$ and for some fixed $c>0$.  Let $t \in L^2_0(\psi)$.
  \begin{enumerate}
  \item We have
    \[
      \langle Pt, t \rangle_\psi \le c \langle Qt, t \rangle_\psi + (1-c) \sigma^2_t.
    \]
  \item Further, assume that $P$ and $Q$ are reversible with respect
    to $\psi$. Then we have
    \begin{enumerate}
          \item Gap $(P) \ge c$ Gap $(Q)$.
    \item Also, if a CLT exists for $t$ under $P$ and $Q$, then
    \[
      v(t, P) \le \frac{v(t, Q)}{c} + \frac{1-c}{c} \sigma^2_t
    \]
    \end{enumerate}
  \end{enumerate} 
\end{theorem}
\begin{remark}
  \cite{zane:2020} presented the results in Theorem~\ref{thm:orde} for
  finite state space Markov chains.  Theorem~\ref{thm:orde} implies
  that if $P$ dominates $Q$ off the diagonal at least $c$ times then
  $P$ is $c$ times more efficient than $Q$ in terms of lag
  autocovariance, spectral gap, and asymptotic variance (ignoring the
  $\sigma^2_t$ term which, as \cite{zane:2020} mentioned, is typically
  much smaller than $v(t, Q)$).
\end{remark}
\section{Analysis of geometric MH algorithms}
  \label{sec:indgeo}

We now use Theorem~\ref{thm:orde} to compare the geometric MH chain
with the MH chain corresponding the base density $f(y|x)$. To that end, let
\begin{equation}
  \label{eq:di}
    d_i(y|x) = \frac{\big(\sqrt{g_i(y|x)/f(y|x)} - \big\langle \sqrt{f(y|x)}, \sqrt{g_i(y|x)} \big\rangle\big)^2}{1 - \big\langle \sqrt{f(y|x)}, \sqrt{g_i(y|x)} \big\rangle^2}
\end{equation}
and
\begin{equation}
  \label{eq:ci}
 c_{i,\epsilon}=\min \big(\inf_{x,y} \{\cos^2(\epsilon \theta_{i,x})+ \sin^2(\epsilon \theta_{i,x}) d_i(y|x)\}, \inf_{x,y} \{\cos^2(\epsilon \theta_{i,y})+ \sin^2(\epsilon \theta_{i,y}) d_i(x|y)\}\big).
\end{equation}
Note that, if $f$ and $g_i$ are symmetric in $(x, y)$ then so is $d_i$
in \eqref{eq:di}. In that case, the two terms inside the minimum in
\eqref{eq:ci} are the same.
\begin{theorem}
  \label{thm:peskmmh}
  Let $c_\epsilon=\sum_{i=1}^k a_i c_{i,\epsilon}$. Let $Q_f$
  be the Mtf of the MH chain with proposal density $f(y|x)$ and
  invariant density $\psi$. Let $P_\phi$ denote the Mtf of the Markov chain underlying
  Algorithm~\ref{alg:mmhmix}. Then, for $\psi$-almost all $x$, we have
  $P_\phi (x, A \setminus \{x\}) \ge c_\epsilon Q_f (x, A \setminus \{x\})$ for all
  $A \in \mathbb{B}(\sX)$.
\end{theorem}

  If $f$ is the proposal density used for an independent MH algorithm
  \citep[][chap. 7.4]{robe:case:2004}, then $f(y|x) =f(y)$. In this
  case, if $g_i$'s also do not involve the current state $x$, then so
  do the proposal density $\phi_{i,\epsilon}$ given in
  \eqref{eq:mixprop}. The proposal density of the geometric MH becomes
  $\phi_{\epsilon} (y) = \sum_{i=1}^k a_i \phi_{i,\epsilon}
  (y)$. Thus, in this case, the geometric MH with a baseline
  independent MH proposal results in an independent MH algorithm, which we refer to as the independent geometric MH. In the following
  corollary, we compare the independent geometric MH chain with its
  baseline MH chain. Let $d_i=\min\{\inf_{x,y} d_i(y|x), \inf_{x,y}d_i(x|y)\}$.
  
\begin{corollary}
  \label{cor:peskindmh}
  Let $c'_\epsilon=\sum_{i=1}^k a_i (\cos^2(\epsilon \theta_i)+
  \sin^2(\epsilon \theta_i) d_i)$. Let $Q$ be the Mtf of the
  independent MH chain with proposal density $f(y)$ and invariant
  density $\psi$. Let $g_i(y|x) =g_i(y)$ for all $i=1,\dots,k$. Let $P$ denote the Mtf of the Markov chain
  underlying Algorithm~\ref{alg:mmhmix}. Then, for $\psi$-almost all
  $x$ we have
  $P_\phi (x, A \setminus \{x\}) \ge c'_\epsilon Q_f (x, A \setminus \{x\})$
  for all $A \in \mathbb{B}(\sX)$.
\end{corollary}
The proof of Corollary~\ref{cor:peskindmh} follows from
Theorem~\ref{thm:peskmmh} as $\theta_{i,x}=\theta_i$ in this case and
thus $c_\epsilon=\sum_{i=1}^k a_i (\cos^2(\epsilon \theta_i)+
  \sin^2(\epsilon \theta_i) d_i) = c'_{\epsilon}$. 
\begin{remark}
  Since $d_i \ge 0$, from Corollary~\ref{cor:peskindmh} we have
  $c'_\epsilon \ge \sum_{i=1}^k a_i \cos^2(\epsilon \theta_i) =
  c^{''}_\epsilon$ say, that is, for the independent geometric MH
  chain,
  $P_\phi (x, A \setminus \{x\}) \ge c^{''}_\epsilon Q_f (x, A \setminus
  \{x\})$ for all $A \in \mathbb{B}(\sX)$.
\end{remark}

Finally, we use the following result from \cite{meng:twee:1996} to
study the convergence properties of the independent geometric MH
algorithm.

\begin{proposition}[Mengersen and Tweedie]
  \label{prop:uniergo}
  The Markov chain underlying the independent geometric MH algorithm 
  is uniformly ergodic if there exists
  $\beta >0$ such that
  \begin{equation}
    \label{eq:min}
    \frac{\phi_{\epsilon}(y)}{\psi(y)} \ge \beta, 
  \end{equation}
  for all $y$ in the support of $\psi$.
  Indeed, under \eqref{eq:min}
\[
\|P_\phi^n(x, \cdot) - \Psi(\cdot)\| \le (1 - \beta)^n,
\]
where $P_\phi^n(x, \cdot)$ is the $n$-step Markov transition function for
the independent geometric MH chain and $\Psi(\cdot)$ is the
probability measure corresponding to the target density $\psi$.

Conversely, if for every $\beta$, there exists a set with positive $\Psi$ measure where
 \eqref{eq:min} does not hold, 
then the manifold MH chain is not even geometrically ergodic.
\end{proposition}
\begin{corollary}
  A sufficient condition for \eqref{eq:min} is that for all
  $i=1,\dots,k$ with $a_i >0$,
  $\phi_{i, \epsilon}(y) \ge \beta \psi(y),$  for all $y$ in the support of $\psi$.
\end{corollary}

\section{Sampling from the geometric MH proposal density}
\label{sec:samph}
The simple mixture representation of the geometric MH proposal
density $\phi_{i, \epsilon}$ in \eqref{eq:mixprop} implies that given
$\epsilon$ and $\theta_{i,x}$, sampling from it can be done by
sampling from $f$ and $h_i$ with probabilities
$\cos^2(\epsilon \theta_{i,x})$ and $\sin^2(\epsilon \theta_{i,x})$,
respectively.  Since $f$ is generally the proposal density of an
uninformed MCMC algorithm already in use in practice, the ability to
sample from \eqref{eq:mixprop} requires successfully sampling from
$h_i$ and computing the mixing weights.

In the special case, when $\epsilon =1$, the mixing coefficients boil
down to $\langle \sqrt{f}, \sqrt{g_i} \rangle^2$ and
$1 - \langle \sqrt{f}, \sqrt{g_i} \rangle^2$. Now $\theta_{i, x}$ as
well as the pdf $h_i$ given in \eqref{eq:defh} involve
$\langle \sqrt{f}, \sqrt{g_i} \rangle$, which is available in closed
form if both $f$ and $g_i$ are normal densities. In particular, if
$f(y)=\phi_d(y;\mu_1,\Sigma_1)$ and $g(y)=\phi_d(y;\mu_2,\Sigma_2)$
where $\phi_d(y;\mu_i,\Sigma_i)$ denotes the probability density
function of the $d-$ dimensional normal distribution with mean vector
$\mu_i$, covariance matrix $\Sigma_i$, and evaluated at $y, i=1,2$,
then
\begin{equation}
\label{eq:normsqdist}
- \log(\langle \sqrt{f},
          \sqrt{g} \rangle) = \frac{1}{8} (\mu_1 -
          \mu_2)^\top\Sigma^{-1}(\mu_1 -
          \mu_2)+\frac{1}{2}\log(|\Sigma|/\sqrt{|\Sigma_1||\Sigma_2|}),
\end{equation}
where $\Sigma = (\Sigma_1+\Sigma_2)/2$. On the other
          hand, if $\langle \sqrt{f}, \sqrt{g_i} \rangle$'s are not
          available in closed form, they can be easily
          estimated. Indeed, since \[ \langle \sqrt{f}, \sqrt{g_i}
          \rangle =\Bigg\langle \sqrt{\frac{f}{g_i}}, g_i \Bigg\rangle =
          \Bigg\langle \sqrt{\frac{g_i}{f}}, f \Bigg\rangle, \] it can
          be consistently estimated using (iid or Markov chain)
          samples from either $g_i$ or $f$ and importance sampling
          methods. Indeed, if $\{X^{(\ell)}\}_{\ell=1}^n$ are
          realizations of a Harris ergodic Markov chain with
          stationary density $f$,  then 
$\sum_{\ell=1}^n            \sqrt{g_i(X^{(\ell)})/f(X^{(\ell)})}/n$ is a consistent
          estimator of $\langle \sqrt{f}, \sqrt{g_i} \rangle$. Note
          that for independent geometric chains, $\langle \sqrt{f},
          \sqrt{g_i}\rangle, i=1,\dots,k$ need to be computed only
          once. The accompanying R package geommc implements the
          importance sampling method if $f$ or $g_i$ is not a normal
          density.

From \eqref{eq:defh} we have
\begin{align}
  \label{eq:rejh}
  h_i(y|x) &= \frac{1}{\sin^2(\theta_{i,x})} \Big(g_i(y|x) + \cos^2(\theta_{i,x}) f(y|x) - 2 \cos(\theta_{i,x})\sqrt{f(y|x)g_i(y|x)} \Big)\nonumber\\
         &\le \frac{1}{\sin^2(\theta_{i,x})} \Big(g_i(y|x) + \cos^2(\theta_{i,x}) f(y|x)\Big) = \frac{1+\cos^2(\theta_{i,x})}{\sin^2(\theta_{i,x})} u_i(y|x),
\end{align}
where
$u_i(y|x) = [g_i(y|x) + \cos^2(\theta_{i,x})
f(y|x)]/[1+\cos^2(\theta_{i,x})]$ is a pdf. Thus, sampling from $h_i$
can be done by a rejection sampler, where a sample from $u_i$ is only
accepted with probability $h_i(x)/[M_{i,x} u_i(x)]$ where
$M_{i,x}= [1+\cos^2(\theta_{i,x})]/\sin^2(\theta_{i,x})$. Also, the
mixture representation of $u_i$ allows sampling from it by drawing
from $f$ and $g_i$ with weights depending only on $\cos(\theta_{i,x})$. Note that
$M_{i,x} \ge 1$, and $M_{i,x} \approx 1$ if
$\theta_{i,x} \approx \pi/2$, which is the case for early iterations
of the Markov chain as the baseline pdf $f$ is generally not `close'
to $g_i$, the approximations of the target density $\psi$. On the
other hand, if $\theta_{i,x}$ is small, for any choice of $\epsilon$,
the mixture weight for $h_i$ in \eqref{eq:mixprop} is small. Thus, in
that case, sampling from $\phi_{i,\epsilon}$ is likely done by
sampling from $f$. Similarly, for small $\epsilon$, the probability
weight for $h_i$ is low in \eqref{eq:mixprop}.  Also, for independent
geometric MH algorithms, $\theta_{i, x}$ needs to be computed
only once, and so is $M_{i, x}$.
\begin{remark}
  The proposed rejection sampler for sampling from $h_i$ and hence
  from the geometric proposal density $\phi_{i, \epsilon}$ in
  \eqref{eq:mixprop} assumes the ability to sample from the
  approximate densities $g_i$'s and to compute the $L^2$ inner product
  $\langle \sqrt{f}, \sqrt{g_i} \rangle, i=1,\dots,k$. In all
  continuous target examples considered in Section~\ref{sec:exam} we
  observe that the choice of Gaussian densities for $g_i$'s leads to
  effective inference using the resulting geometric MH algorithms. As
  mentioned before, in this case,
  $\langle \sqrt{f}, \sqrt{g_i} \rangle$ is available in closed form,
  if $f$ is also a Gaussian density.
\end{remark}

\section{Examples}
\label{sec:exam}
In this section, we illustrate the performance of the proposed
geometric MH algorithms in the context of six different examples. In
Examples~\ref{sec:exampnormtar} and \ref{sec:exampcauchtar}, the
independent and random walk MH chains, respectively, are known to
suffer from slow mixing, and the proposed framework is shown to turn
these chains, which are not even geometrically ergodic into uniformly
ergodic MCMC algorithms. Also, we use these two examples to
demonstrate and compare the performance of the general rejection
sampling scheme described in Section~\ref{sec:samph} with examples
specific tuned rejection sampling for the geometric MH proposal
density \eqref{eq:mixprop}. Example~\ref{sec:exampmixture} considers a
bivariate multimodal target corresponding to mixtures of normals where
the random walk MH schemes are known to get stuck in a local mode. As
in Examples~\ref{sec:exampnormtar} and \ref{sec:exampcauchtar}, the
geometric MH scheme successfully turns these poor mixing random walk
kernels into algorithms that efficiently move between the local
modes. In Example~\ref{sec:exampsixmode}, a geometric MH-within-Gibbs
sampler successfully finds all six modes of a non-normal, multi-modal
target density where the RW-within-Gibbs algorithm stays trapped in a
local mode, resulting in erroneous
inferences. Example~\ref{sec:examplogistic} considers the widely used
Bayesian logistic models where the geometric MH algorithm is compared
with a variety of other MH schemes, including the manifold MALA. In
Example~\ref{sec:exampsglmm}, these different algorithms are
constructed and compared in a MH-within-Gibbs framework in the context
of analyzing spatial GLMMs. Finally, Example~\ref{sec:exampvarsel} involves discrete target
probability mass functions (pmfs) arising from a popular Bayesian
variable selection model. The superiority of the geometric MH
algorithm over other MH schemes is demonstrated using extensive
ultra-high dimensional simulation examples as well as a real dataset
from a genome-wide association study (GWAS) with close to a million
markers.

\begin{example}[Independent MH for N (0, 1) target]
\label{sec:exampnormtar}
Let the target density $\psi$ be the standard normal N $(0, 1)$
  density. If the proposal density $f$ of the independent MH algorithm
   is N $(1, 1)$, then
  \begin{equation}
    \label{eq:indmhnorm}
    \frac{f(x)}{\psi(x)} = \exp(x - 0.5) \rightarrow 0 \;\;\mbox{as}\;\; x \rightarrow - \infty 
  \end{equation}
and the acceptance ratio 
\[
\alpha(x, y) = \mbox{min} \Big\{\frac{\psi(y) f(x)}{\psi(x) f(y)}, 1 \Big\}=  \mbox{min} \{ \exp(x-y), 1\}.
\]
So the moves to the right are possibly rejected, but moves to the left
are always accepted. Indeed, the independent MH chain leaves the sets
$(-\infty, n]$ more and more slowly once it enters
them. Figures~\ref{fig:normtrace} (a) and (b) show the trace plots of
this independent MH chain starting at -5 and -10,
respectively. The plots reveal the slow mixing
of the chain, whereas when started at -10, the chain has failed to
move.

Next, we consider our proposed Algorithm~\ref{alg:mmhmix}. With the
baseline density $f$ same as the independent MH algorithm, we form the
proposal density
$\phi(x) = \cos^2(\epsilon \theta) f(x) + \sin^2(\epsilon \theta)
h(x)$ by perturbing $f$ in the direction of the target density
$g = \psi$. From \eqref{eq:normsqdist} we have
$\langle \sqrt{f}, \sqrt{g} \rangle = \exp(-1/8)$. Thus, $\theta = \cos^{-1}(\exp(-1/8))$ and
\begin{equation}
  \label{eq:normg}
    h(x) = \frac{1}{1-\exp(-1/4)} (\sqrt{g(x)}- \exp(-1/8)\sqrt{f(x)})^2.
\end{equation}
  Now, $h(x)$ can be expressed as either \[ h(x) =
\frac{g(x)}{1-\exp(-1/4)} \Bigg(1-
\exp(-1/8)\sqrt{\frac{f(x)}{g(x)}}\Bigg)^2, \] or \[ h(x) =
\frac{f(x)}{1-\exp(-1/4)} \Bigg(\sqrt{\frac{g(x)}{f(x)}}-
\exp(-1/8)\Bigg)^2.  \] It can be shown that \[ \sup_{x \in (-\infty,
3/4]} \Bigg(1- \exp(-1/8)\sqrt{\frac{f(x)}{g(x)}}\Bigg)^2 = 1 \] and
\[ \sup_{x \in [3/4, \infty)} \Bigg(\sqrt{\frac{g(x)}{f(x)}}-
  \exp(-1/8)\Bigg)^2 = \exp(-1/4).  \] Thus,
$h(x) \le \tilde{M} \tilde{u}(x)$, where
$\tilde{M} = [\Phi(3/4)+\exp(-1/4) \Phi(1/4)]/[1-\exp(-1/4)]$ and the
density
\[ \tilde{u}(x) = \frac{g(x)I(x \le3/4) + f(x)\exp(-1/4)I(x
    >3/4)}{\Phi(3/4)+\exp(-1/4) \Phi(1/4)}.  \] Thus, we can sample
from $h$ by first drawing $X \sim \tilde{u}$ and then accepting it
with probability $h(x)/[\tilde{M} \tilde{u}(x)]$. For sampling from
$h$, the rejection sampler based on $\tilde{u}$ mentioned here is
better than the general method mentioned based on the representation
\eqref{eq:rejh} as, in this example
$M \equiv M_{i, x} = 8.042 > 5.604 = \tilde{M}$. On the other hand,
for $\epsilon =0.5, \sin^2(\epsilon \theta) =0.0588$. Thus, for
sampling from \eqref{eq:mixprop}, $h$ is sampled less than 6\% of the
time.  We run the geometric MH algorithm for 1000 iterations with
$\epsilon =0.5$ using the general method \eqref{eq:rejh} starting at
$-30$. From Figures~\ref{fig:normtrace} (c), we see that, unlike the
independent MH, the proposed manifold MH algorithm successfully
quickly moves to the modal region of the target, starting even at -30.
      
\begin{figure*}
  \begin{minipage}[b]{0.24\linewidth}
    \includegraphics[width=\linewidth]{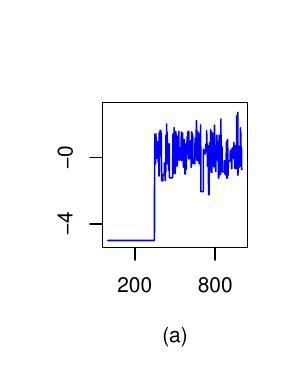}
  \end{minipage}
  \begin{minipage}[b]{0.24\linewidth}
    \includegraphics[width=\linewidth]{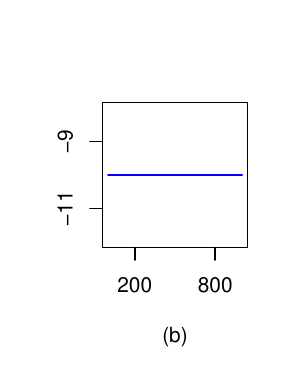}
  \end{minipage}
  \begin{minipage}[t]{0.24\linewidth}
    \includegraphics[width=\linewidth]{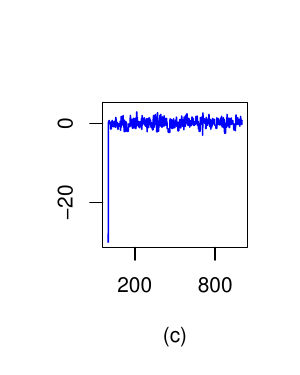}
  \end{minipage}
    \begin{minipage}[b]{0.24\linewidth}
    \includegraphics[width=\linewidth]{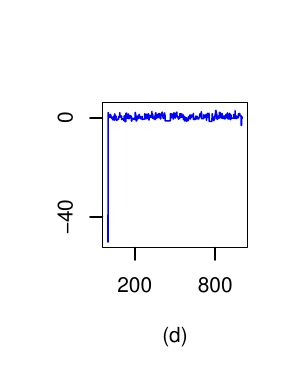}
  \end{minipage}
  \caption{Trace plots for one thousand iterations of the independence
    Metropolis chains (a) starting at -5, (b) starting at -10, and of
    the geometric MH chains (c) with $g = \psi$ starting at -30, and
    (d) $g$ as N $(0, 30^2)$ density starting at -50 for the standard
    normal target. }
\label{fig:normtrace}
\end{figure*}

From \eqref{eq:indmhnorm} and \cite{meng:twee:1996}, we know that the
independent MH chain is not geometrically ergodic. On the other hand,
uniform ergodicity of the geometric MH algorithm follows from the
following Lemma.
\begin{lemma}
  \label{lemm:uenorm}
  The geometric MH chain corresponding to N (1,1) base density is uniformly ergodic.
\end{lemma}
 Finally, we consider an independent
geometric MH algorithm with a non-informative choice of $g$, namely,
the density of N $(0, 30^2)$. Figures~\ref{fig:normtrace} (d) shows
that like the other geometric chain, this also moves to the modal
region in one step.
\end{example}

\begin{example}[Random walk MH for Cauchy (0, 1) target]
\label{sec:exampcauchtar}
  Let the target density $\psi$ be the Cauchy $(0, 1)$
  density. We first consider the RWM chain with the increment $I \sim$ N $(0, 1)$. From
  \cite{jarn:twee:2003} we know this RWM chain is not
  geometrically ergodic. We ran this chain for ten thousand iterations, and the leftmost plot of Figure~\ref{fig:cauchyacf} reveals the slow
  mixing of the chain. A heavier-tailed proposal in RWM for
  the Cauchy target can lead to a polynomial rate of convergence
  \citep{jarn:robe:2007}. We ran the RWM chain with 
  $t_2$ proposal for ten thousand iterations and the second from the left plot of
  Figure~\ref{fig:cauchyacf} shows the high lag autocorrelations of this chain.

Next, we consider our proposed Algorithm~\ref{alg:mmhmix} with
independent base density $f= t_\kappa,$ the $t$ density with degrees
of freedom $\kappa$. With the baseline density $f = t_{\kappa}$,
we form the proposal density $\phi(x) = \cos^2(\epsilon \theta) f(x) +
\sin^2(\epsilon \theta) h(x)$ by perturbing $f$ in the direction of
the target density $g = \psi$. Since \[ h(x) = \frac{1}{1-\langle
\sqrt{f}, \sqrt{g} \rangle^2} \Bigg(\frac{1}{\sqrt{\pi(1+x^2)}}-
\langle \sqrt{f}, \sqrt{g} \rangle \sqrt{\frac{\Gamma(\frac{\kappa
+1}{2})}{\sqrt{\kappa}\Gamma(\frac{\kappa}{2})}} \sqrt{\frac{1+
x^2}{(1 +x^2/\kappa)^{(\kappa+1)/2}}}\Bigg)^2, \]
with degrees of freedom $\kappa =2$, we have
  \begin{equation}
   \label{eq:cauchyh}
  h(x) = \frac{1}{1-\langle \sqrt{f}, \sqrt{g} \rangle^2} \frac{1}{\pi(1+x^2)} \Bigg(1- \langle \sqrt{f}, \sqrt{g} \rangle \Big(\frac{\pi}{2}\Big)^{1/4} 2^{3/4} \sqrt{\frac{1+ x^2}{(2 +x^2)^{3/2}}}\Bigg)^2.
\end{equation}
Now \[ \sup_x \Bigg(1- \langle \sqrt{f}, \sqrt{g} \rangle
\Big(\frac{\pi}{2}\Big)^{1/4} 2^{3/4} \sqrt{\frac{1+ x^2}{(2
    +x^2)^{3/2}}}\Bigg)^2 =1. \]
Thus, we can sample from $h$ by first drawing
$X \sim $ Cauchy (0, 1) and then accepting it with probability
$(1-\langle \sqrt{f}, \sqrt{g} \rangle^2)h(x)/g(x)$. Again, in this
case, $1/(1-\langle \sqrt{f}, \sqrt{g} \rangle^2) = 25.538 < M = M_{i,
x} = 50.077$ implying that the sampler based on the representation
\eqref{eq:rejh} results in lower acceptance rates than the rejection
sampler based on Cauchy $(0, 1)$. On the other hand, with $\epsilon
=0.5, \sin^2(\epsilon \theta) = 0.0099,$ implying that a sample from
$h$ is drawn less than 1\% of the time while sampling from $\phi$. We ran the independent geometric MH algorithm with
$f=t_2$ for 10,000 iterations with $\epsilon =0.5$ using
\eqref{eq:rejh} to sample from $h$. From plot (c) in
Figure~\ref{fig:cauchyacf} we see that the independent geometric MH
algorithm results in much lower autocorrelations than the random walk
MH chains. By
Proposition~\ref{prop:uniergo} we know that the independent MH chain
with $f=t_2$ proposal is not geometrically ergodic as $\min_x
f(x)/\psi(x) =0$. On the other hand, since
\[
  \min_x \frac{\phi(x)}{\psi(x)} = \sin^2(\epsilon\theta) \min_x \Bigg(1- \langle \sqrt{f}, \sqrt{g} \rangle \Big(\frac{\pi}{2}\Big)^{1/4} 2^{3/4} \sqrt{\frac{1+ x^2}{(2 +x^2)^{3/2}}}\Bigg)^2 >0
  \]
  for any $\epsilon>0$, by Proposition~\ref{prop:uniergo}, the
  independent geometric MH algorithm is uniformly ergodic with $t_2$
  baseline density.  Next, we ran the geometric MH algorithm with
  baseline proposal density $f=N(x,1)$ (that is, the normal RWM
  proposal density) for 10,000 iterations with $\epsilon =0.5$ (plot
  (d) in Figure~\ref{fig:cauchyacf}). Unlike the other MH chains,
  autocorrelations of the geometric MH chains drop down to
  (practically) zero by five lags, revealing their fast mixing
  properties. Finally, we consider versions of the two geometric MH
  algorithms used here with $g$ as N (0, 900) density. From the Figure
  S1 given in the supplement, we see that the geometric algorithms
  with this choice of $g$ also result in rapidly declining
  autocorrelations values.

\begin{figure*}
  \begin{minipage}[b]{0.24\linewidth}
    \includegraphics[width=\linewidth]{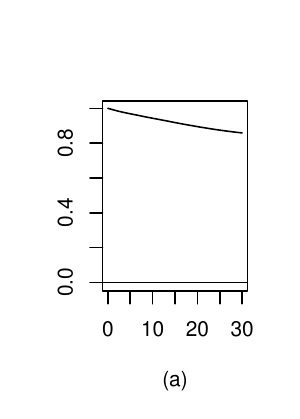}
  \end{minipage}
  \begin{minipage}[b]{0.24\linewidth}
    \includegraphics[width=\linewidth]{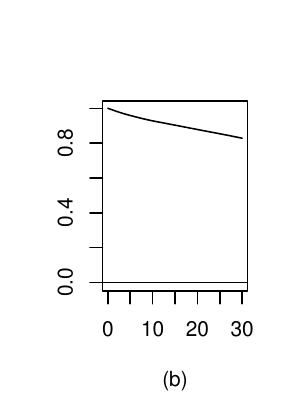}
  \end{minipage}
    \begin{minipage}[b]{0.24\linewidth}
    \includegraphics[width=\linewidth]{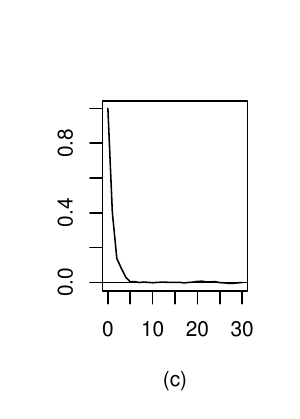}
  \end{minipage}
      \begin{minipage}[b]{0.24\linewidth}
    \includegraphics[width=\linewidth]{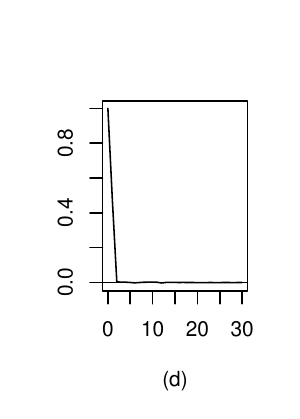}
  \end{minipage}

 \caption{Autocorrelation function plots of the (a) random walk chain with normal proposal, (b) random walk chain with $t_2$ proposal, (c) independent geometric MH chain with $t_2$ baseline density and (d) the geometric MH chain with $N(x,1)$ baseline density where $x$ is the current state of the Markov chain. }
\label{fig:cauchyacf}
\end{figure*}
\end{example}

\begin{example}[Mixture of bivariate normals]
\label{sec:exampmixture}
  Suppose the target density $\psi$ is
  $\psi(x) =0.5 \phi_2(x; \mu_1, \Sigma_1) + 0.5 \phi_2(x; \mu_2,
  \Sigma_2)$, where  $\mu_1=(0,0)^\top$, $\mu_2=(10,10)^\top$,
  $\Sigma_1 = I$ and $\Sigma_2 = 2I$.

  We ran the RWM chain with the normal proposal density with
  covariance matrix $2I$ started at $(5, 5)^\top$ for $100,000$
  iterations. The estimated acceptance rate is 42.14\%. The first row
  in Table~\ref{tab:mixest} provides the estimates of the means of the
  two coordinates, mean square Euclidean jump distance(MSJD),
  effective sample size (ESS) for the two coordinates, the
  multivariate ESS (mESS) and its time normalized value (mESS/sec)
  based on this RWM chain. The MSJD for a Markov chain
  $\{X^{(n)}\}_{n \ge 1}$ is defined as
  $\sum_{i=2}^{n+1} \|X^{(i)} -X^{(i-1)}\|^2/n$. We use the R package
  mcmcse \citep{R:mcmcse} for computing ESS and mESS values. The mean
  estimates are far from the true value $(5, 5)$. The reason is that
  the random chain failed to move out of the local mode at $(0, 0)$
  even after $100,000$ iterations. The left panel of
  Figure~\ref{fig:mixts} shows only the first 1000 steps of the chain.

  We then ran the proposed geometric random walk chain with the same
  baseline density as the random walk MH chain for $100,000$
  iterations initialized at $(5, 5)^\top$. In this case, we took
  $k=2$, $g_1(y) \equiv \phi_2(y; \mu_1, \Sigma_1)$ and
  $g_2 (y) \equiv \phi_2(y; \mu_2, \Sigma_2)$. Thus,
  $\mathcal{G} = \{g_1, g_2\}$ is a set of two unimodal densities
  centered at the different local modes. Note that since
  $f(y|x) = \phi_2(y ; x, 2I)$,
  $\langle \sqrt{f}, \sqrt{g_i} \rangle, i=1, 2$ are available in
  closed form. From the middle panel of Figure~\ref{fig:mixts}, which
  shows the first 1000 steps of the geometric chain, we see that by
  combining the localized steps of the RW with the global moves, the
  geometric MCMC chain can successfully move back and forth between
  the two modes and simultaneously explore the modal regions. Table~\ref{tab:mixest} shows that the geometric chain
  (GMC1(RW)) results in higher ESS values and much higher MSJD,
  demonstrating better mixing and successful reduction of the RW behavior.

  Next, we consider the geometric random walk chain with the same
  baseline density $f$ as the RWM chain but $k=1$ and
  $g=\psi$. In this case, $\langle \sqrt{f}, \sqrt{g} \rangle$ is not available in
  closed form and needs to be estimated. We estimate it by importance
  sampling using samples from $f$. The results from $100,000$
  iterations of this chain (GMC2(RW)) initialized at $(5, 5)^\top$ are
  given in Table~\ref{tab:mixest}. From
  these results, we see that GMC2(RW) performs better than the
  GMC1(RW) chain.

  Finally, we ran a geometric Markov chain (GMC3(RW)) with the same base
random walk proposal density $\phi_2(y;x,2I),$ but with a
`non-informative' choice of $g$, namely $g(y)=\phi_2(y;0,30^2I)$, a
normal density centered at $0$ with covariance matrix $900I$. From
Figure~\ref{fig:mixts}, we see that even with this choice of a diffuse
density for $g$, the geometric MH chain successfully moves between the
two local modes, resulting in estimates of the means close to their
true values. On the other hand, unlike GMC1 and GMC2, GMC3 takes more
time to move between the modes, resulting in lower MSJD values than
GMC1 and GMC2. Note that the mESS for a Markov chain sample of size
$n$ is defined as $n \sqrt{|\Lambda|/|\Sigma|}$ where $\Lambda$ is the
sample covariance matrix and $\Sigma$ is an estimate of the Monte
Carlo covariance matrix. Since the RW chain is stuck in a local mode,
it is fooled into treating the target distribution as unimodal, and it
results in small values of $|\Lambda|$ and $|\Sigma|$ leading to a
higher value of mESS than that of GMC3 \citep{roy:2020}.
\begin{figure*}
    \begin{minipage}[b]{0.24\linewidth}
    \includegraphics[width=\linewidth]{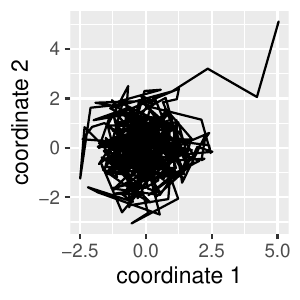}
  \end{minipage}
  \begin{minipage}[b]{0.24\linewidth}
    \includegraphics[width=\linewidth]{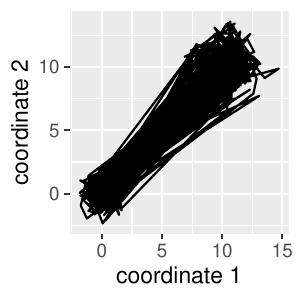}
  \end{minipage}
    \begin{minipage}[b]{0.24\linewidth}
    \includegraphics[width=\linewidth]{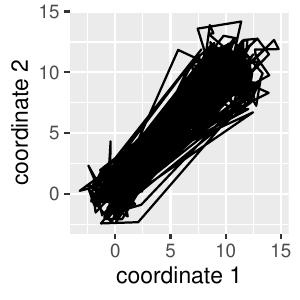}
  \end{minipage}
      \begin{minipage}[b]{0.24\linewidth}
    \includegraphics[width=\linewidth]{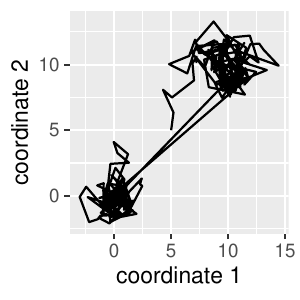}
  \end{minipage}

  \caption{Illustration of trajectories for a mixture of two
    normals. The trajectory of the first 1000 iterations of different chains. The random walk chain (left plot) entirely lies
    in one mode, whereas the first geometric MH chain (second plot from the left), the
    second geometric MH chain (third plot from the left) and the
    third geometric MH chain (rightmost plot) with $N_2(x,2I)$ baseline density successfully move
    between the modes where $x$ is the current state of the Markov
    chain.}
 \label{fig:mixts}
  \end{figure*}
Next, we compare the three chains in terms of
  computational time. For completing $100,000$ iterations, the RW
  chain took around 10 seconds on an old Intel i7-870 2.93 GHz machine
  running Windows 10 with 16 GB RAM, whereas the GMC1(RW), GMC2(RW) and GMC3(RW)
  chains took a little less than two minutes, about 20 minutes, and about one minute,
  respectively. The extra time required in GMC2(RW) is due to
  computation of $\langle \sqrt{f}, \sqrt{g} \rangle$ by importance sampling in
  every iteration of the Markov chain.
  \begin{center}
\begin{table*}[h]
  \caption{Results for different samplers for the mixture model example}
  \centering
\begin{tabular}{lcccccc}
\hline\hline  Sampler  & $E(X_1)$ & $E(X_2)$ & MSJD & mESS &  mESS/sec &ESS \\
 \hline
  RW& -0.004& -0.004&0.916 & 11911 & 964&(12224,11437)  \\
  GMC1(RW)&\phantom{-} 6.729 &\phantom{-} 6.721 &29.586 & 17520& 156&(15132,15549)\\
  GMC2(RW)&\phantom{-} 5.025 &\phantom{-} 5.035 &33.352 & 23740 & 17&(20371,18756)\\
  GMC3(RW)&\phantom{-} 4.675 &\phantom{-} 4.705 &0.957 & 1152 &15 &(16578,17133)\\
\hline
\end{tabular}
\label{tab:mixest}
\end{table*}
\end{center}
\end{example}
\begin{example}[A target distribution with six modes]
  \label{sec:exampsixmode}
The following target density \eqref{eq:six_tar} is from \cite{lem:chen:lavi:2009} 
\begin{equation}\label{eq:six_tar}
\psi(x_1, x_2)\propto\exp\bigg(\frac{-x_1^2}{2}\bigg)\exp\bigg(-\frac{((csc~ x_2)^5 - x_1)^2}{2}\bigg),\; -10 \le x_1, x_2 \le 10.
\end{equation}
The marginal densities of $X_1$ and
$X_2$, known up to a normalizing constant, given in
Figure~\ref{fig:six_marg}, clearly show that the target distribution
has six well separated modes.
\begin{figure*}
    \begin{minipage}[b]{0.5\linewidth}
    \includegraphics[width=\linewidth]{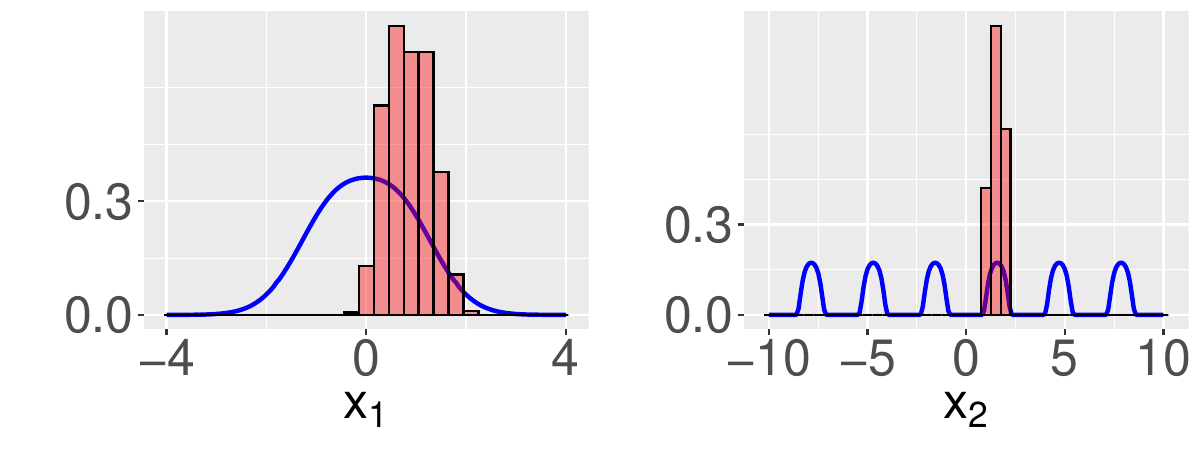}
  \end{minipage}
      \begin{minipage}[b]{0.5\linewidth}
    \includegraphics[width=\linewidth]{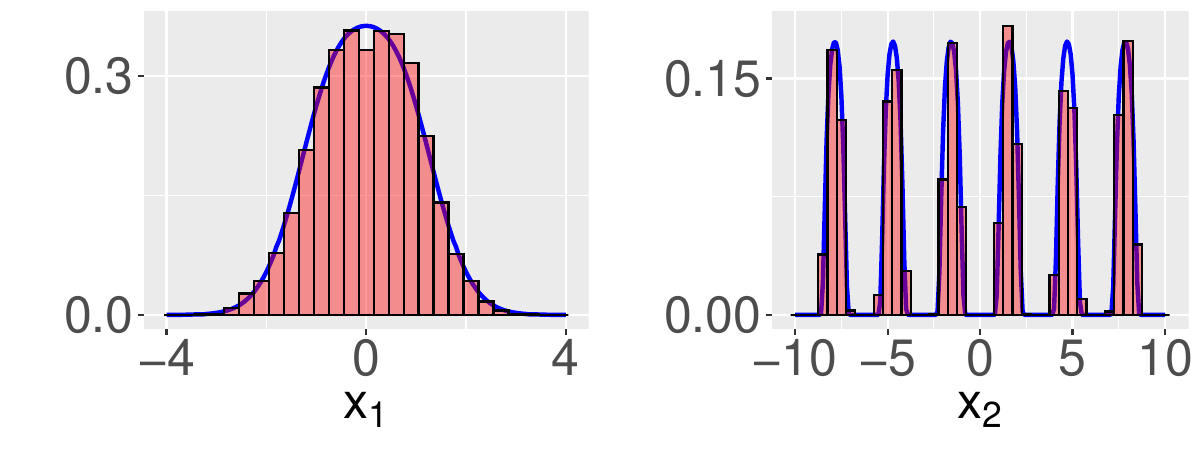}
  \end{minipage}
  \caption{Desired and attained marginals from 100,000 iterations of two MCMC samplers for the six-mode target example.
    The first and second plots from the left show marginal histograms from the RW-within-Gibbs sampler demonstrating that the chain entirely lies in one mode. Whereas the histograms (the third and fourth plots from the left) obtained from
    the geometric RW-within-Gibbs sampler show that the chain successfully moves between the modes.}
 \label{fig:six_marg}
  \end{figure*}
  We first consider a RW-within-Gibbs sampler. In particular,
  following \cite{lem:chen:lavi:2009}, this sampler iteratively
  samples from $\psi(x_1|x_2)$ and
  $\psi(x_2|x_1)$ using a RWM algorithm with the normal proposal with
  standard deviation
  $0.1$. We ran this chain, denoted by RW-w-Gibbs, for 100,000
  iterations starting at $(0.05, 1.5)$. From Figure~\ref{fig:six_marg} we see that the chain is
  trapped in one of the modes, and has failed to converge to the
  target density \eqref{eq:six_tar}.

  Next, we ran another algorithm where iteratively sampling from the
  conditional densities $\psi(x_1|x_2)$ and
  $\psi(x_2|x_1)$ are done using the proposed geometric RW
  sampler. That is, in every iteration of the Gibbs sampler the two
  RWM steps for the conditional densities are replaced by one
  transition from their corresponding GMC (RW) steps with the same baseline RW
  densities. While the base densities are the same as the proposals in
  RW-w-Gibbs, a `non-informative' choice, namely a normal density with
  mean zero and variance 900, is made for the
  $g$ density. This sampler, denoted by GMC(RW)-w-Gibbs, was also run
  for 100,000 iterations starting at $(0.05,
  1.5)$. From Figure~\ref{fig:six_marg} we see that the geometric
  MH-within-Gibbs algorithm has successfully accomplished the two
  important tasks: find all the modes, and move between them. Indeed
  as in Example~\ref{sec:exampmixture}, the `global' moves of the
  geometric proposal complements the localized steps of the RW by
  moving between the modes.

  The better mixing of the GMC(RW)-w-Gibbs chain is corroborated by
  much higher MSJD values given in Table~\ref{tab:sixest}. The
  geometric sampler results in estimates of the marginal means closer
  to the true values, whereas the estimates obtained from the
  RW-w-Gibbs chain are far off. As mentioned in the previous example,
  the ESS estimates for the RW-w-Gibbs chain can't be
  trusted. Finally, we consider an algorithm where only the RW step
  for $\psi(x_2|x_1)$ is replaced by its geometric variant mentioned
  above. From the histogram given in Figure S2 in the supplement, we
  see that this chain, denoted by GMC($X_2$)-w-Gibbs, also
  successfully moves between the modes although compared to when
  geometric sampler is used to sample from $\psi(x_1|x_2)$,
  GMC($X_2$)-w-Gibbs results in poorer estimate of the marginal
  density of $X_1$. We also ran the Markov chain obtained by replacing
  only the RW step for $\psi(x_1|x_2)$ with its geometric invariant. This chain 
  remained stuck (not shown here) in a local mode for the whole 100,000
  iterations. 
  
    \begin{center}
\begin{table*}[h]
  \caption{Results for different samplers for the six-modal target example}
  \centering
\begin{tabular}{lcccccc}
\hline\hline  Sampler  & $E(X_1)$ & $E(X_2)$ & MSJD & mESS &mESS/sec & ESS \\
 \hline
  RW-w-Gibbs& \phantom{-}0.681&\phantom{-} 1.569&0.017 & 749&68&(530,1740)  \\
    GMC(RW)-w-Gibbs&-0.172 &\phantom{-}-0.027 &1.450 & 337&5 &(90,973)\\
  GMC($X_2$)-w-Gibbs&\phantom{-}0.013 &0.033 &1.469 & 756&6 &(373, 998)\\
\hline
\end{tabular}
\label{tab:sixest}
\end{table*}
\end{center}

\end{example}
\begin{example}[Bayesian logistic model]
\label{sec:examplogistic}
  We consider the Binary logistic regression model. Suppose
  $z = (z_1, z_2, \dots, z_m)$ are $m$ independent observations where
  $z_i$ is either $0$ or $1$ and assume that
  $ Z_i \sim \mbox{Ber}(\xi_i)\ ,$ with
  $\log(\xi_i/[1-\xi_i]) = w_i^\top \beta $ where $w_i$'s, $i = 1,2,\dots,m$
  are the $p \times 1$ covariate vectors and $\beta$ is the
  $p \times 1$ vector of regression coefficients. We consider the
  $p$-variate normal prior $N_p(\mu_0, \Sigma_0)$ with mean $\mu_0$
  and covariance matrix $\Sigma_0$. Thus, the target posterior density is
  $\psi(\beta) \propto \bigg[\prod^m_{i=1} \frac{\exp(z_iw^\top_i\beta)}{1+ \exp(w^\top_i\beta)}\bigg] \times \phi_p(\beta ; \mu_0, \Sigma_0).$

  We consider the Pima Indian data set \citep{ripl:1996}.  A
  population of women who were at least 21 years old of Pima Indian
  heritage, and living near Phoenix, Arizona, was tested for diabetes,
  according to World Health Organization criteria. The data were
  collected by the US National Institute of Diabetes and Digestive and
  Kidney Diseases. We used the $m=532$ complete records selected from a
  larger data set, with the binary observation denoting the presence
  or absence of diabetes, and $p=8$ covariates consisting of an
  intercept term and the following seven predictors: the number of
  pregnancies, plasma glucose concentration in an oral glucose
  tolerance test, diastolic blood pressure (mm Hg), triceps skin fold
  thickness (mm), body mass index (weight in
  kg/$\mbox{(height in m)}^2$), diabetes pedigree function, and age
  (in years). In the supplement we consider the German Credit dataset
  for which $p=21$ and $m=1000$.

  We analyze the Pima Indian data set by fitting a Bayesian logistic
  regression model using the RWM, independent MH algorithms, the MALA, the MMALA, and
  their geometric MH variants.
  We take $\mu_0 = 0_8$, a vector of zeros, and $\Sigma_0 = 10^3 I_8$
  for the prior density of $\beta$.  We ran the algorithms for 100,000
  iterations all started at $\mu_0$.  Let $\hat{\beta}$ be the MLE of
  $\beta$. We obtain $\hat{\beta}$ by fitting the glm function of R with the logit link.  Let
  $\hat{\Sigma} = (-\nabla^2 \log
  \psi(\beta))^{-1}|_{\beta=\hat{\beta}}$, the generalized observed
  Fisher information matrix. For the RWM, we use the normal
  proposal with covariance matrix $\Sigma_f=0.3 \hat{\Sigma}$ to get
  an acceptance rate of around 50\%.  For the independent MH, we use a
  normal proposal with mean $\hat{\beta}$ and covariance matrix
  $\hat{\Sigma}$ resulting in an acceptance rate of around 80\%. For
  constructing MALA and MMALA, we need the first and higher-order
  derivatives of the log target density $\psi$. 
  Indeed, the proposal
  density of MALA is $\phi(\beta'; \beta+h\nabla\log \psi(\beta)/2, hI_p)$ for some step-size $h$ and
  $\beta$ is the current state of the Markov chain.  It turns out that
  $\nabla \log \psi(\beta) = W^{\top}(z - \xi)- \Sigma_0^{-1}(\beta
  -\mu_0), \nabla^2 \log \psi(\beta) = -W^{\top}\Lambda W -
  \Sigma_0^{-1}$ and
  $ \frac{\partial \nabla^2 \log \psi(\beta)}{\partial \beta_j} =
  -W^{\top}\Gamma^j X$, where $W$ is the $m \times p$ matrix of
  covariates, $\xi = (\xi_1,\dots,\xi_m)$, 
  $\Lambda$ is the $m\times m$ diagonal matrix with $i$th diagonal
  element $\xi_i(1- \xi_i), i=1,\dots,m$, and $\Gamma^j$ is the
  $m\times m$ diagonal matrix with $i$th diagonal element
  $\xi_i(1- \xi_i)(1- 2\xi_i)w_{ij}, i=1,\dots,m$. For MALA and MMALA
  we take $h=0.01$ and $h= 2$, respectively resulting in around 50\%
  acceptance rates.

  For the geometric variants of the RW and independent
  MH algorithms mentioned above, we take $k=1$ and
  $g(\beta) = \phi_p(\beta; \hat{\beta}, \hat{\Sigma})$.  For the
  independent geometric chain (GMC(Ind)), we first consider
  $f(\beta)= \phi(\beta; \mu_0, \Sigma_f)$, although later we discuss
  another choice. In this case,
  $\langle \sqrt{f}, \sqrt{g} \rangle = \langle \sqrt{\phi(\beta;
    \mu_0, \Sigma_f)}, \sqrt{g} \rangle \approx \pi/2$, and thus
  $M \equiv M_{i, \beta} \approx 1$ where $M_{i, \beta}$ is defined in
  Section~\ref{sec:samph}. On the other hand, for the RW geometric
  chain (GMC(RW)), $\langle \sqrt{f}, \sqrt{g} \rangle$ needs to be
  computed in every iteration. Since the chains are started at
  $\mu_0$, for the first iteration, the value of
  $\langle \sqrt{f}, \sqrt{g} \rangle$ will be the same as that of the
  independent geometric chain. If the mean of $f$ is $ \hat{\beta}$,
  then
  $\langle \sqrt{\phi(\beta; \hat{\beta}, \Sigma_f)}, \sqrt{g} \rangle
  = 1.042$, and in that case $M =1.682$ giving around 60\% acceptance
  probability for the sampling algorithm described in
  Section~\ref{sec:samph}. Similarly, we consider geometric variants
  of MALA (GMC(MALA)) and MMALA (GMC(MMALA)). Table~\ref{tab:pimaess}
  provides ESS, mESS, their time normalized values, and MSJD values
  for the different samplers. For ESS, we provide the minimum, median,
  and maximum of the eight values corresponding to $p=8$
  covariates. Table~\ref{tab:pimaacf} shows the first eight
  autocorrelations for the function $\beta^\top W^\top W \beta$ for
  the different Markov chains. The function $\beta^\top W \beta$ is a
  natural choice as it is used as the drift function to prove the
  geometric ergodicity of some Gibbs samplers for some Bayesian binary
  regression models \cite[see e.g][]{roy:hobe:2007}. The numerical
  results corroborate Theorems~\ref{thm:orde}-\ref{thm:peskmmh} as
  the geometric MCMC algorithms lead to better performance over the
  base Markov chains. Note that the base kernel of GMC(Ind) is not the
  same as the proposal density of the independent MH (Ind). The
  improvement obtained by the geometric perturbation over an uninformed base
  kernel like RW is much higher than that over a more informed kernel
  like MALA and MMALA. GMC(MMALA)
  has superior performance over the other algorithms in terms of ESS
  and autocorrelation, closely followed by the independent MH
  algorithm. On the other hand, in terms of the time-normalized
  values, the independent MH algorithm with the proposal $g$ beats all
  other algorithms. Also, the time-normalized ESS values of the
  geometric MCMC algorithms based on simpler base kernels like RW are
  about three times higher than those of manifold MALA chains, and this
  domination is even bigger (about twenty times) for the German Credit
  dataset with more ($21$) covariates. 
  Next, we consider the independent MH algorithm with a proposal
  $f_1(\beta) = \phi_p(\beta; \mu_0, \num{e-6}I_8)$. This algorithm
  (Ind ($f_1$)) results in small ESS and MSJD values and large
  autocorrelations. Thus, the independent MH algorithm's performance
  suffers greatly with the proposal density change. On the other hand,
  the performance of the GMC (Ind) does not vary much when
  $f(\beta)= \phi(\beta; \mu_0, \Sigma_f)$ is replaced with
  $f_1(\beta)$. 
  For the geometric variants of the algorithms, we did not choose the
  step size or the variance of the baseline densities, optimizing
  their empirical performance; we used $\epsilon =0.5$ and other
  values the same as the non-geometric chains. By changing $\epsilon$,
  geometric algorithms' acceptance rate and performance can greatly
  vary. For example, the GMC(RW) chain in Table~\ref{tab:pimaess} with
  $\epsilon =0.5$ has 62\% acceptance rate, whereas for
  $\epsilon =0.1$ and $0.9$, it results in 45\% and 83\% acceptance
  rates, respectively. From Tables S1-S2 given in the supplement, we
  see that the comparative performance of the 10 MCMC samplers remains
  similar for the German Credit data set. However, as mentioned
  before, some improvements obtained by the geometric MCMC algorithms
  are even larger for the German Credit data set.
  
  \begin{center}
\begin{table*}[h]
  \caption{Multivariate ESS and ESS (minimum, median, maximum) and their time normalized values for different samplers for the Pima Indian data}
  \centering
  \begin{tabular}{cccccc}
\hline\hline  Sampler  & mESS  & ESS & mESS/sec  & ESS/sec  &MSJD \\
 \hline
  RW&2765&(2405,2833,3155) &45&(39,46,51) &0.019  \\
    GMC(RW)&22460&(18094,21210,23873) &158&(127,149,167) &0.123\\
  Ind&56782&(42972,59874,62279) &745 &(564,786,818)& 0.263 \\
  GMC(Ind)& 21405&(19266,21078,23126)&180&(162,177,195)&0.131 \\
  Ind($f_1$)&77 &(32,72,83)&0.97&(0,1,1) &2.7e-7\\
  GMC(Ind ($f_1$))& 21406&(19263,21078,23125)&174&(157,172,188)&0.131\\
  MALA &10948 &(4953,7473,12245) &120 &(54,82,134)&0.033\\
  GMC(MALA)&28386 &(20963,22638,26727)&101 &(74,80,95)&0.133\\
  MMALA & 35032 &(33070,34381,36491) &57 &(54,56,59)&0.136\\
  GMC(MMALA)&58160 &(51354,59985,61494)&15&(13,15,16)&0.248\\
\hline
\end{tabular}
\label{tab:pimaess}
\end{table*}
\end{center}

  \begin{center}
\begin{table*}[h]
  \caption{First eight autocorrelations for different samplers for the Pima Indian data}
  \centering
\begin{tabular}{ccccccccc}
\hline\hline  Sampler  & lag1  & lag2 & lag3  & lag4 &lag5 &lag6 &lag7 &lag8\\
 \hline
  RW&0.941& 0.886 &0.835 &0.787 &0.742 &0.699 &0.658& 0.620  \\
    GMC(RW)&0.663 &0.452 & 0.314& 0.223& 0.162& 0.122& 0.096& 0.078\\
  Ind&0.313& 0.143& 0.082& 0.059& 0.044& 0.037& 0.028& 0.021\\
  GMC(Ind)&0.657& 0.446& 0.310& 0.221& 0.164& 0.126& 0.098& 0.079 \\
  Ind($f_1$)& 0.997& 0.995& 0.993& 0.991& 0.989& 0.987& 0.985& 0.984\\
  GMC(Ind ($f_1$))& 0.657& 0.446& 0.310& 0.221& 0.164& 0.126& 0.098& 0.079\\
  MALA & 0.784& 0.630& 0.514& 0.426& 0.356& 0.300& 0.252& 0.209\\
  GMC(MALA)&0.603 & 0.379 &0.243& 0.161& 0.108& 0.073& 0.051& 0.038\\
  MMALA &0.440& 0.238& 0.142& 0.097& 0.068& 0.046& 0.032& 0.021\\
  GMC(MMALA)&0.342& 0.130& 0.053 &0.0267& 0.018& 0.016& 0.008& 0.006\\
\hline
\end{tabular}
\label{tab:pimaacf}
\end{table*}
\end{center}  
\end{example}

\begin{example}[Bayesian spatial GLMM]
  \label{sec:exampsglmm}  In spatial GLMMs, conditional on the underlying latent spatial
  Gaussian process $\{S(l), l \in \mathcal{L}\}$, the observations
  $z=(z(l_1), \dots, z(l_m))$ at observed locations
  $(l_1, \dots, l_m)$ are assumed to be independent random variables
  following the exponential family with their (conditional) means related
  to $S(l)$'s via a link function. Let $S_i \equiv S(l_i)$ denotes the
  value of the underlying process at location $l_i, i=1,\dots,m$ and
  $S=(S_1,S_2,\dots,S_m)^{\top}$. Denoting $z(l_i)$ by $z_i$, for
  analyzing spatial count data, we assume that conditional on $S$,
  $z_i | s_i \stackrel{\mathrm{ind}}{\sim} \mathrm{Poisson}
  (t_i\lambda_i)$ where $\lambda_i$ is the rate parameter, the link
  function is $s_i=\log(\lambda_i)$ and $t_i$ is some `exposure'
  variable to the event, for example, the number of hours of
  operation, or the area for the $i$th location, $i=1,\dots,m$. Assume
  that $S$ has the mean $X\beta$, where $X$ is the $m\times p$ covariates
  matrix and $\beta\in\mathbb{R}^p$ is the regression parameter. Also,
  let the covariance matrix $\Sigma$ of $S$ be formed from an
  exponential correlation function that is,
  $\text{cov}(S(l),S(l'))=\sigma^2 \exp\{-\theta\| l-l'\|+\omega \}$,
  where $\sigma^2, \theta, \omega$ are the partial sill, range and
  relative nugget parameters, respectively, and $\|l-l'\|$ is the
  Euclidean distance between $l$ and $l'$. There are other choices of
  parametric correlation functions available in the literature
  \citep{digg:ribe:chri:2003}.

  We consider a Bayesian analysis and assume the normal prior on
  $\beta$ with mean $\mu_0$ and covariance matrix
  $\sigma^2\Sigma_{0}$, that is, apriori
  $\beta|\sigma^2 \sim N_p (\mu_0, \sigma^2 \Sigma_0)$. The parameters
  $(\sigma^2, \theta, \omega)$ are assumed apriori independent with
  $\sigma^2 \sim \text{Inverse Gamma}(\alpha_1,\gamma_1), \theta \sim
  \text{Inverse Gamma}(\alpha_2,\gamma_2),$ and
  $\omega\sim \text{Inverse Gamma}(\alpha_3,\gamma_3)$ for some known
  hyper parameter values of $\alpha_i$ and $\gamma_i$, $i=1,2,3$. We
  run different MH-within-Gibbs algorithms to sample from the
  posterior density $\psi(S,\beta,\sigma^2,\theta,\omega|z)$. The
  joint density $\psi(S,\beta,\sigma^2,\theta,\omega|z)$ and its
  conditionals are derived in the supplement. From these derivations,
  we see that the full conditional densities of $\beta$ and $\sigma^2$
  are normal and inverse gamma densities, respectively. Conversely,
  the conditional densities of $S$, $\theta$, and $\omega$ are not
  standard. Among these, the density of $S$ is high-dimensional, with
  the dimension being the same as the number of observations, $m$. We
  use RWM steps for sampling from the conditional densities of
  $\log(\theta)$ and $\log(\omega)$. For sampling from the conditional
  density of $S$, we consider various MH algorithms, particularly RWM,
  MALA, and MMALA, and their geometric variants. The supplement
  provides the first and higher-order derivatives of the logarithm of
  the conditional density of $S$ required for constructing the MALA
  and MMALA. Also, the order in which different variables are sampled
  in the Gibbs samplers are given in the supplement.

  We consider the Rongelap island dataset, which consists of the
  measurements of $\gamma$-ray counts $z_i$ observed during $t_i$
  seconds at $i$th coordinate on the Rongelap island, $i=1,\ldots,m$,
  with $m=157$. This data set was analyzed by
  \cite{digg:tawn:moye:1998, evan:roy:2019} among others, using a
  Poisson spatial model. For this dataset, $p=1$ and we assume
  $\mu_0 = 0$, and $\Sigma_0 = 100$.  For the priors on
  $\sigma^2, \theta$ and $\omega$, we assume
  $\alpha_1=\alpha_2=\alpha_3=2.04$ and $\beta_1=\beta_2=\beta_3=2.08$
  setting their prior means and variances at 2 and 100, respectively.

  For the RWM for the conditional densities of $\log(\theta)$ and
  $\log(\omega)$, we use normal proposals with variances 0.04 and 1,
  respectively to get acceptance rates of around 50\%. For sampling
  from the conditional density of $S$ using RWM we use the normal
  proposal with covariance matrix $3*10^{-5} I_{157}$ and for MALA and
  MMALA we take $h=0.0001$ and $h= 0.5$, respectively resulting in
  around 50\% acceptance rates. We ran the algorithms for 100,000
  iterations all started at
  $s=\log(z), \beta=\mu_0, \sigma^2=\theta=\omega=2$ on a Linux server
  equipped with 128 AMD EPYC 7542 CPU cores and 1 TB of RAM.

  For the geometric variants of the algorithms mentioned above, for
  sampling from the conditional density of $S$, we take $k=1$ and
  $g(s) = \phi_{157}(s; s_{\text{mean}}, G)$, where $s_{\text{mean}}$
  is the estimate of the posterior mean of $S$ obtained from the above
  mentioned MH-within-Gibbs sampler with the RWM step for $S$ and
  $G\equiv G_{\sigma^2,\theta,\omega}= (-\nabla^2 \log
  \psi(s|\beta,\sigma^2,\theta,\omega,z))^{-1}|_{s=s_{\text{mean}}}$
  evaluated at the values of $\sigma^2,\theta,$ and $\omega$ in that
  iteration. Note that $G$ and hence the density $g$ change in every
  iteration, unlike in the previous examples.

  Table~\ref{tab:rongess} provides ESS, mESS, their time normalized
  values, and MSJD values for $S$ based on the different samplers. For
  ESS, we provide the minimum, median, and maximum of the 157 values
  corresponding to $m=157$ coordinates of $S$. Table~\ref{tab:rongacf}
  shows the first eight autocorrelations for the function $S^\top S$
  for the different Markov chains. From
  Tables~\ref{tab:rongess}-\ref{tab:rongacf} we see that GMC(MMALA)
  has superior performance over the other algorithms in terms of
  autocorrelation, MSJD, median ESS and time normalized ESS, mESS
  values. As in Example~\ref{sec:examplogistic}, the improvement of
  the geometric algorithm over an uninformed base kernel like RW is
  much higher than that over a more informed kernel like MALA and
  MMALA. Also, as in Example~\ref{sec:examplogistic}, for the
  geometric variants of the algorithms, we used $\epsilon =0.5$, and
  other values remained the same as the non-geometric chains.
  
  \begin{center}
\begin{table*}[h]
  \caption{Multivariate ESS and ESS (minimum, median, maximum) and their time normalized values for different samplers for the Rongelap island data}
  \centering
  \begin{tabular}{cccccc}
\hline\hline  Sampler  & mESS  & ESS & mESS/sec  & ESS/sec  &MSJD \\
 \hline
  RW&651&(7,110,650) &0.157&(0.002,0.027,0.157) &0.000  \\
    GMC(RW)&30077&(21535,24487,27973) &5.084&(3.640,4.139,4.729) &0.067\\
  MALA &11573 &(98,4055,49329) &2.414 &(0.020,0.846,10.290)&0.010\\
  GMC(MALA)&35927 &(23312,27997,61011)&5.580 &(3.621,4.349,9.476)&0.072\\
  MMALA & 9952&(6581,8207,11121)&2.021 &(1.337,1.667,2.259)&0.026 \\
  GMC(MMALA)&34640 &(25600,29244,34472)&7.065&(5.221,5.964,7.030)&0.075\\
\hline
\end{tabular}
\label{tab:rongess}
\end{table*}
\end{center}

    \begin{center}
\begin{table*}[h]
  \caption{First eight autocorrelations for different samplers for the Rongelap island data}
  \centering
\begin{tabular}{ccccccccc}
\hline\hline  Sampler  & lag1  & lag2 & lag3  & lag4 &lag5 &lag6 &lag7 &lag8\\
 \hline
  RW&0.998 &0.996 &0.994 &0.992 &0.990 &0.989 &0.987 &0.985\\
    GMC(RW)&0.586 &0.349 &0.212 &0.129 &0.088 &0.061 &0.038&0.023\\
  MALA & 0.926 &0.862 &0.803 &0.751 &0.704 &0.659 &0.617 &0.580\\
  GMC(MALA)&0.553 &0.312 &0.177 &0.104 &0.066 &0.042 &0.026 &0.016\\
  MMALA &0.845 &0.716 &0.608 &0.517 &0.439 &0.374 &0.321 &0.276\\
  GMC(MMALA)&0.541&0.300 & 0.169&  0.097&  0.060&  0.037 & 0.023 & 0.014\\
\hline
\end{tabular}
\label{tab:rongacf}
\end{table*}
\end{center}  
\end{example}

\begin{example}[Bayesian variable selection]
\label{sec:exampvarsel}
  We now consider the so-called variable selection problem, where we
  have a $m \times 1$ vector of response values
  $z=(z_1, \ldots, z_m)$, a $m \times p$ design matrix
  $W = (W_1, \ldots, W_p)$ with each column of $W$ representing a
  potential predictor and the goal is to identify the set of all
  important covariates which have non-negligible effects on the
  response $z$. In a typical GWAS, the number
  of markers, $p$ far exceeds the number of observations $m$, although
  only a few of these variables are believed to be associated with the
  response. Here, we consider a popular approach to variable selection
  \citep{mitc:beau:1988,geor:mccu:1993, geor:mccu:1997, nari:he:2014,
    li:dutt:roy:2023} based on a Bayesian hierarchical model mentioned
  below.

  Let $\beta$ denote the $p$
  dimensional vector of regression coefficients, $\gamma$ denote a subset of $\{1,2,\dots,p\}$ and the cardinality of
  $\gamma$ be denoted by $|\gamma|$. Corresponding to a given model
  $\gamma$, let $W_{\gamma}$ denote
  the $n \times |\gamma|$ sub-matrix of $W$ and $\beta_{\gamma}$
  denotes the $|\gamma|$ dimensional sub-vector of $\beta$ . Let $1_{m}$ denotes a $m$-vector of 1's. The Bayesian hierarchical
regression model we consider is given by
\begin{subequations}\label{eq:litdiff}
\begin{align}
    z |\beta,\beta_0,\sigma^2, \gamma &\sim N_{m}\left(1_{m} \beta_{0}+W_{\gamma} \beta_{\gamma}, \sigma^{2} I\right), \label{subeq:regModel}\\
  \beta_{j} | \beta_0, \sigma^2, \gamma &\stackrel{\text { ind }} \sim N\left(0, \frac{\gamma_{j}}{\lambda} \sigma^{2}\right) \text { for } j=1, \ldots, p,\label{subeq:priorBeta}\\
  \left(\beta_{0}, \sigma^2\right)|\gamma &\sim 1 / \sigma^{2}, \label{subeq:priorInterceptVariance}\\
                                                \gamma|\omega &\sim \omega^{|\gamma|}(1-\omega)^{p-|\gamma|}\label{subeq:priorGamma}.
\end{align}
\end{subequations}
In the model  \eqref{eq:litdiff},  \eqref{subeq:regModel} indicates that conditional on the
parameters, each $\gamma$ corresponds to a Gaussian linear regression
model $z = \beta_01_m + W_\gamma\beta_\gamma + \epsilon$ where the
residual vector $\epsilon \sim N_{m}(0,\sigma^2 I).$ Given $\gamma$, a
popular non-informative prior is set for $(\beta_0, \sigma^2)$ in
\eqref{subeq:priorInterceptVariance} and a conjugate independent
normal prior is used on $\beta$ in
\eqref{subeq:priorBeta} with the common parameter $\lambda > 0$ 
controlling the precision of the prior. Following the common practice,
here, we assume that the covariate matrix $W$ is scaled. Note that, if
a covariate is not included in the model, the prior on the
corresponding regression coefficient degenerates at zero. The prior of
$\gamma$ in \eqref{subeq:priorGamma} is obtained by assuming
independent Bernoulli distribution for each indicator variable
indicating the presence or absence of variables corresponding to the
model and $\omega \in (0, 1)$ is the prior inclusion probability of
each predictor. The hyperparameters $\lambda$ and $\omega$ are
assumed known \citep[see][for appropriate choices of these
parameters.]{nari:he:2014, li:dutt:roy:2023}.

  It is possible to analytically integrate out $\beta_0,$
  $\beta_\gamma$ and $\sigma^2$ from the hierarchical model
  \eqref{eq:litdiff}, and the marginal posterior pmf of $\gamma$ is
  given by
  \begin{equation}
    \label{eq:marginalGamma}
    \psi(\gamma | z) \propto \lambda^{|\gamma|/2} |A_\gamma|^{-1/2} R_\gamma^{-(m-1)/2} \omega^{|\gamma|}(1-\omega)^{p-|\gamma|},
  \end{equation}
  where $A_\gamma = W_{\gamma}^{\top} W_{\gamma}+\lambda I,$
  $|A_\gamma|$ is the determinant of
  $A_\gamma, R_{\gamma} = \tilde{z}^{\top}
  \tilde{z}-\tilde{z}^{\top}W_{\gamma}A_\gamma^{-1}W_\gamma^\top\tilde{z}$
  is the ridge residual sum of squares, and $\tilde{z}=z-\bar{z} 1_{m}$ with $\bar{z} = \sum_{i=1}^m z_i/m.$ The density
  \eqref{eq:marginalGamma} is usually explored by MCMC sampling and
  several MH and Gibbs algorithms have been proposed in the literature
  \citep[see e.g][]{geor:mccu:1997, guan:step:2011,
    yang:wain:jord:2016, grif:latu:2021, zane:robe:2019, zhou:yang:2022,
    lian:livi:2022}. The proposal densities of MH algorithms for
  \eqref{eq:marginalGamma} are generally mixtures of three types of
  local moves, namely ``addition'', ``deletion'' and ``swap''. In
  order to describe these proposals, for a given model $\gamma$, let
  $\mathcal{N}(\gamma) = \gamma^{+} \cup \gamma^{\circ} \cup
  \gamma^{-}$ denote a neighborhood of $\gamma$, where $\gamma^{+}$ is
  an ``addition'' set containing all the models with one of the
  $p-|\gamma| $ remaining covariates added to the current model
  $\gamma$, $\gamma^{-}$is a ``deletion'' set obtained by removing one
  variable from $\gamma;$ and $\gamma^{\circ}$ is a ``swap'' set
  containing the models with one of the variables from $\gamma$
  replaced by one variable from $\gamma^c.$ The proposal densities of
  the MH chains are of the form
  \begin{equation}
    \label{eq:vsmhprop}
    f(\gamma'|\gamma) = \frac{b^{+}(\gamma) I_{\gamma^{+}}(\gamma')}{p-|\gamma|}+\frac{b^{-}(\gamma)I_{\gamma^{-}}(\gamma')}{|\gamma|} + \frac{b^{\circ}(\gamma)I_{\gamma^{\circ}}(\gamma')}{|\gamma|(p-|\gamma|)},
  \end{equation}
  where $b^{+}(\gamma), b^{-}(\gamma), b^{\circ}(\gamma)$ are
  non-negative constants summing to 1. When $b^{+}(\gamma),
  b^{-}(\gamma), b^{\circ}(\gamma)$ are all constants independent of
  $\gamma$, \cite{zhou:yang:2022} refer to the resulting MH algorithm
  as asymmetric RW. In their simulation examples,
  \cite{zhou:yang:2022} sets $b^{+}(\gamma) = b^{-}(\gamma)=0.4,
  b^{\circ}(\gamma)=0.2$. \cite{yang:wain:jord:2016} set
  $b^{+}(\gamma) = (p-|\gamma|)/2p, b^{-}(\gamma) = |\gamma|/2p$ and
  $b^{\circ}(\gamma) = 1/2$, and since in this case,
  $f(\gamma'|\gamma)= f(\gamma|\gamma')$ for all $\gamma,
  \gamma'$, the resulting MH algorithm is called the symmetric
  RW. \cite{yang:wain:jord:2016} establish rapid mixing (mixing
  time is polynomial in $m$ and $p$) of the symmetric RW algorithm,
  but in high dimensional examples, the RW algorithms can suffer
  from slow convergence, and their efficiency can be improved by using
  informative proposals. Indeed, motivated by \cite{zane:2020}, where
  variable selection was not discussed explicitly, a couple of other
  informative proposals have recently been constructed, for example,
  the tempered Gibbs sampler of \cite{zane:robe:2019}, the adaptively
  scaled individual adaptation proposal of \cite{grif:latu:2021}, and
  the Locally Informed and Thresholded proposal distribution of
  \cite{zhou:yang:2022}.

We now consider our proposed geometric MH algorithm with
$k=1$, the base density \eqref{eq:vsmhprop} and
\begin{equation}
\label{eq:vsaptar}
g(\gamma'|\gamma)  = 
\begin{cases}
\frac{\psi(\gamma|z)}{c_{\gamma}} & \text{ if } \gamma' \in \mathcal{N}(\gamma),\\
0, & \text{otherwise},
\end{cases}
\end{equation}
where
$c_{\gamma} = \sum_{\gamma' \in \mathcal{N}(\gamma)}
\psi(\gamma|z)$. For implementing this proposed algorithm, we need to
efficiently compute $\psi(\gamma|z)$ for all
$\gamma' \in \mathcal{N}(\gamma)$. Note that \cite{zhou:yang:2022}
consider only the `addition' and `deletion' moves in their simulation
examples, greatly reducing the computational burden. Here, we use the
fast Cholesky updates of \cite{li:dutt:roy:2023} for rapidly computing
$\psi(\gamma|z)$ for all $\gamma' \in \mathcal{N}(\gamma)$. Once
$g(\gamma'|\gamma), \gamma' \in \mathcal{N}(\gamma)$ are computed, the
inner product $\langle \sqrt{f}, \sqrt{g} \rangle$ is available.
 Furthermore, both \cite{yang:wain:jord:2016} and \cite{zhou:yang:2022} use
\pcite{zell:1986} $g-$prior on $\beta$ and a different prior on
$\gamma$. The $g-$prior, although a popular alternative to
the independent normal prior \eqref{subeq:priorBeta}, it requires all
$m \times q$ sub-matrices of $W$ have full column rank for $q \le m-1$ and the
support of the prior on $\gamma$ is restricted to models of size at most $m-1$.

We now perform extensive simulation studies and compare our geometric
MH algorithm to the RW algorithms. In particular, we consider the
symmetric RW (RW1), the asymmetric RW with $b^{+}(\gamma) =
b^{-}(\gamma)=0.4, b^{\circ}(\gamma)=0.2$ (RW2), the two geometric MH
algorithms with the base density \eqref{eq:vsmhprop} corresponding to
 RW1 and RW2, denoted by GMC1 and GMC2, respectively. Our numerical studies are
conducted in the following five different simulation settings.

\noindent{\bf Independent predictors:} In this example, following \cite{li:dutt:roy:2023}, entries of $W$ are
generated independently from $N(0,1)$. The coefficients are specified
as $\beta_1=0.5, \beta_2=0.75, \beta_3=1, \beta_4=1.25, \beta_5=1.5,$
and $\beta_j=0, \forall j > 5.$

\noindent{\bf Compound symmetry:} This example is taken from Example 2 in
\citet{wang:leng:2016}. The rows of $W$ are generated independently
from $N_p\left(0,(1-\rho)I_p + \rho1_p1_p^\top\right)$ where we take
$\rho = 0.6$. The regression coefficients are set as $\beta_j = 5$ for
$j = 1, \ldots, 5$ and $\beta_j=0$ otherwise.

\noindent{\bf Auto-regressive correlation:} Following Example 2 in
\citet{wang:leng:2016}, $W_j = \rho W_{j-1} + (1-\rho^2)^{1/2}b_j,$
for $1 \leq j \leq p,$ where $W_0$ and $b_j$ ($1\leq j \leq p)$ are
iid $\sim N_m(0,I_m).$ Following \cite{li:dutt:roy:2023}, we use $\rho=0.6$ and set the regression
coefficients as $\beta_1 = 3$, $\beta_4=1.5$, $\beta_7=2$ and
$\beta_j=0$ for $j \not \in \{1,4,7\}$.

\noindent{\bf Factor models:} Following \citet{wang:leng:2016} and
\cite{li:dutt:roy:2023}, we first generate a $p\times 2$ factor matrix
$F$ whose entries are iid standard normal. Then the rows of $W$ are
independently generated from $N_p(0,FF^\top + I_p).$ The regression
coefficients are set to be the same as in compound symmetry example.

\noindent{\bf Extreme correlation:} Following \citet{wang:leng:2016}, in this
challenging example, we first simulate $b_j$ , $j=1, \ldots, p$ and
$t_j$, $j=1, \ldots, 5$ independently from $N_m(0, I_m)$. Then the
covariates are generated as $W_j=(b_j+t_j)/\sqrt{2}$ for $j=1, \ldots,
5$ and $W_j=(b_j+\sum_{i=1}^5 t_i)/2$ for $j = 6, \ldots, p$. As in
\cite{li:dutt:roy:2023}, we set $\beta_j = 5$ for $j = 1, \ldots, 5$
and $\beta_j=0$ for $j = 6, \dots, p$. Thus, the correlation between
the response and the unimportant covariates is around $2.5/\sqrt{3}$
times larger than that between the response and the true covariates,
making it difficult to identify the important covariates.

Our simulation experiments are conducted using 100 simulated pairs of
training and testing datasets. For each simulation setting, we set
$p=10000$ and $m=400$ for both training and testing data
sets. Following \cite{li:dutt:roy:2023}, we choose $w$ and $\lambda$
to be $\sqrt{m}/p$ and $m/p^2$, respectively. 
The error variance $\sigma^2$ is set by assuming different theoretical
$R^2$ values. While the results for $R^2 = 90\%$ are provided in
Table~\ref{tab:rsqpt9}, the supplement contains the results for $R^2 =
60\%$ and $R^2 = 75\%$. All results are based on 100 iterations of the
GMC chains and 50,000 iterations of the RW chains with all chains
started at the null model.

To evaluate the performance of the MCMC algorithms,
we consider several metrics that we describe now. Following
\cite{zhou:yang:2022}, let $\hat{\gamma}_{\text{max}}$ be the model
with the largest posterior probability that has been sampled by any of
the four algorithms. If an algorithm has never sampled
$\hat{\gamma}_{\text{max}}$, the run is considered as a failure. Also,
let $\hat{\gamma}_{\text{med}}$ be the median probability model, that is, $\hat{\gamma}_{\text{med}}$ is the
set of variables with estimated marginal inclusion probability (MIP)
above 0.5 \citep{barb:berg:2004}. We compute (1) Number of runs (Success) out of 100
repetitions when $\hat{\gamma}_{\text{max}}$ is sampled (2) median
number of iterations ($N_{\text{success}}$) needed to sample
$\hat{\gamma}_{\text{max}}$ among the successful runs (3) median time
in seconds (Time) to reach $N_{\text{success}}$ among the successful
runs (4) mean squared prediction error (MSPE) based on
$\hat{\gamma}_{\text{med}}$ for testing data (5) mean squared error
(MSE$_{\bm{\beta}}$) between the estimated regression coefficients
corresponding to $\hat{\gamma}_{\text{med}}$ and the true coefficients
(6) average model size (size), which is calculated as the average
number of predictors included in $\hat{\gamma}_{\text{med}}$ overall
the replications (7) coverage probability (coverage) which is defined
as the proportion of times $\hat{\gamma}_{\text{med}}$ contains the
true model (8) false discovery rate (FDR) for $\hat{\gamma}_{\text{med}}$ (9) false negative rate
(FNR) for $\hat{\gamma}_{\text{med}}$ and (10) the Jaccard index, which is defined as the size of the
intersection divided by the size of the union of
$\hat{\gamma}_{\text{med}}$ and the true model. All computations for these simulation examples are
done on the machine mentioned in Example~\ref{sec:exampmixture}.

\begin{table}
\begin{center}
\small\addtolength{\tabcolsep}{-2pt}
\caption{Results for different samplers for the Bayesian variable selection example. Success, Coverage, FDR,
 FNR and Jaccard Index are reported in percentages.}
\label{tab:rsqpt9}
\begin{tabular}{lrrrrrrrrrr}
  \hline\hline
        &Success&$N_{\text{success}}$&Time 
       &MSPE     & MSE$_\beta$
        &\shortstack{Model \\ size} &Coverage
       & FDR
       & FNR
        &\shortstack{Jaccard \\ Index} \\
          \hline
          & \multicolumn{10}{c}{Independent design}\\
  \hline
  RW1  &64&32458 &126.89& 1.803&1.192&3.42&19&0.0 & 31.6&68.4\\
  GMC1 &100&9&0.92&0.636&0.008 &5.00 & 100&0.0&0.0 &100.0\\
    RW2 &49& 31849&120.45& 2.224&1.570&3.20&11&0.0&36.0&64.0\\
  GMC2 &100&10&0.97&0.636&0.008 &5.00 & 100&0.0&0.0 & 100.0\\
	\hline
 & \multicolumn{10}{c}{Compound symmetry design with $ r= 0.6$}\\
  \hline
  RW1  &56 &36126& 213.50& 56.998& 22.450 &4.39 &50 &0.7 &12.8 &86.8\\
  GMC1 &100 &    9&   0.92 &48.149 & 1.366 &5.00 &100 &0 &0.0& 100.0  \\
    RW2 &62& 27286& 141.17& 70.738& 51.577& 3.84 &28& 3.3& 26.0& 72.9 \\
  GMC2 &100 &   10&  0.98& 48.149&  1.366& 5.00& 100& 0& 0.0& 100.0\\
	\hline
         & \multicolumn{10}{c}{Autoregressive correlation design with $ r= 0.6$}\\
  \hline
  RW1  & 62 & 31092 &126.58& 3.433& 1.855& 2.79 & 64& 6.7 &13.3 &83.6 \\                  
  GMC1 &  100 &    6 &  0.47& 2.148& 0.021& 3.01 &100 &0.2 &0.0& 99.8  \\
    RW2 &  82 &24836&  94.07& 4.790& 3.866& 2.65 &45 &11.3& 22.3& 73.6 \\
  GMC2 & 100&   6 &  0.44& 2.148 &0.021& 3.01 &100& 0.2 &0.0 &99.8\\
        \hline
         & \multicolumn{10}{c}{Factor model design}\\
  \hline
  RW1  & 54& 34870 &175.73 & 79.107 &30.563 &4.11 &33& 1.8 & 19.6 & 79.5\\
  GMC1 &99  &  11  & 1.12&  43.014 & 1.338 &4.96 &99& 0.0& 0.8 &99.2\\   
    RW2 &  61 & 35321& 162.96& 102.566 & 49.133& 3.52 & 13 & 2.2 & 31.4 & 67.6\\
  GMC2 &  100&    10  & 0.97&  42.005&  0.328 & 5.00& 100& 0.0 & 0.0& 100.0\\
        \hline
         & \multicolumn{10}{c}{Extreme correlation design}\\
  \hline
  RW1  &   47 & 35036& 171.90& 39.850& 27.131& 4.02& 36& 1.4 &20.8 &78.5\\
 GMC1 &  100&    12&   1.22& 16.968 & 6.563& 4.85& 96 &3.5 & 3.8  &96.2\\
 RW2 & 58& 31982& 145.68& 48.720 &37.798& 3.70& 19 &3.1 &28.4 &70.3\\
  GMC2 & 100 &   11&  1.07& 14.182&  0.178 &5.00 &100 &0.0 &0.0 &100.0\\
        \hline
\hline
\end{tabular}
\end{center}
\end{table}
From Table~\ref{tab:rsqpt9}, we see that the proposed geometric algorithms
successfully find the model $\hat{\gamma}_{\text{max}}$ in all 100
repetitions across all five simulation settings. Whereas the RW
algorithms failed to find $\hat{\gamma}_{\text{max}}$ around $40 \%$
of the time, and the failure rate could be as high as $53 \%$. Also,
across all three different values of $R^2$, we see that the proposed
geometric algorithms always hit $\hat{\gamma}_{\text{max}}$ much
faster than the RW algorithms. Indeed, remarkably, starting from the
null model, the median number of iterations ($N_{\text{success}}$) to
reach the model $\hat{\gamma}_{\text{max}}$ for the informative MH
algorithms is always less than 15, except for the extreme correlation
design with $R^2 =75 \%$ when $N_{\text{success}} = 25$ for GMC1. On
the other hand, among the successful runs, RW algorithms generally
required about 30,000 iterations to find
$\hat{\gamma}_{\text{max}}$. From Table~\ref{tab:rsqpt9}, we see that
the median wall time needed for the geometric MH algorithms to
generate $N_{\text{success}}$ samples is less than 1.3 seconds in all
five scenarios, whereas, among the successful runs, the average time
to reach $\hat{\gamma}_{\text{max}}$ for the RW algorithms was as high
as 3.5 minutes. Also, model averaging with 100 samples of the
geometric MH algorithms results in much better performance than for the RW
algorithms with 50,000 iterations. Indeed, the median probability model
$\hat{\gamma}_{\text{med}}$ obtained from the informed MH algorithms
is generally the true model leading to about 100\% coverage. Also, the
geometric MH algorithms resulted in larger Jaccard index values
and smaller MSPE and MSE$_{\bm{\beta}}$ values. Also, although the
symmetric RW generally outperformed the asymmetric RW, the geometric
MH algorithms based on either of these RW algorithms resulted in
similar performance. Thus, the GMC algorithms outperform the RW algorithms in terms of finding the maximum a posteriori model
$\hat{\gamma}_{\text{max}}$, as well as leading to better model fitting
and prediction accuracies based on Bayesian model averaging.

{\bf Real data analysis:} We now consider an ultra-high dimensional real
dataset and analyze it by fitting \eqref{eq:litdiff} using the
proposed geometric MH algorithm. This maize shoot apical meristem
(SAM) dataset was generated by \cite{leib:li:2015}. The maize SAM is a
small pool of stem cells that generate all the above-ground organs of
maize plants. \cite{leib:li:2015} showed that SAM size is correlated
with a variety of agronomically important adult traits such as
flowering time, stem size, and leaf node number. In
\cite{leib:li:2015}, a diverse panel of 369 maize inbred lines was
considered, and close to 1.2 million single nucleotide polymorphisms
(SNPs) were used to study the SAM volume. After removing duplicates
and SNPs with minor allele frequency (MAF) less than 5\%, we end up
with $p = 810,396$ markers, and the response $z$ is the log of the SAM
volume for $m=369$ varieties. The inbred varieties are
bi-allelic, and we store the marker information in a sparse format by
coding the minor alleles by one and the major alleles by zero.

We first ran the GMC1 chain with $\lambda = m/p^2= 5.61865e-$10 and
$w = \sqrt{m}/p= 2.37037e-$5 (the default choices for $\lambda$ and
$w$ in the function {\it geomc.vs} in the accompanying R package
geommc) for $N=100$ iterations starting from the null model. Only one
variable resulted in with estimated MIP above 0.5. Indeed,
$\hat{\gamma}_{\text{med}} = 1_{83878775}$. The $R^2$ value based on
fitting $\hat{\gamma}_{\text{med}}$ on the response $z$ is
20.07\%. Next, we ran the GMC1 chain for 100 iterations starting at
the null model, but this time following \cite{li:dutt:roy:2023}, we
took $\lambda = \sqrt{m}=19.20937$ (high shrinkage) and a higher value
for $w = m/p=0.00046$. This time, $\hat{\gamma}_{\text{med}}$
contained four variables with
$\hat{\gamma}_{\text{med}}= (1_{83878775}, 2_{175541357},
3_{226820323}, 8_{27127525})$. In addition to the estimated median
model $\hat{\gamma}_{\text{med}}$, we also consider a weighted average
model (WAM) $\hat{\gamma}_{\text{wam}}$. For the unique MCMC samples
$\{\gamma^{(n)}\}_{n= 1}^N$ we assign the weights
$\text{wt}_n = \psi(\gamma^{(n)}| z)/\sum_{i=1}^N
\psi(\gamma^{(i)}|z)$ according to the marginal posterior pmf
\eqref{eq:marginalGamma}. Then, the approximate marginal inclusion
probability for the $j$th variable is computed as
$\hat\pi_j = \sum_{n=1}^N \text{wt}_n \mathbb{I}(\gamma^{(n)}_j = 1)$
and define the WAM as the model containing variables $j$ with
$\hat\pi_j > 0.5.$ Based on the 100 iterations of the GMC1 chain,
$\hat{\gamma}_{\text{wam}}$ consisted of five variables with
$\hat{\gamma}_{\text{wam}} = (1_{83878775}, 2_{79769999},
2_{175541357}, 8_{27127525}, 10_{10606917})$. The $R^2$ values for
$\hat{\gamma}_{\text{med}}$ and $\hat{\gamma}_{\text{wam}}$ are
39.10\% and 43.93\%, respectively. We analyzed the real dataset on the
Linux server mentioned in Example~\ref{sec:exampsglmm}, where it took
5.34 minutes to complete 100 iterations of GMC1 with
$\lambda = \sqrt{m}$ and $w = m/p$. With these values of
$(\lambda, w)$, we repeated 100 iterations of GMC1 50 times, each time
with a different seed. The models $\hat{\gamma}_{\text{med}}$ and
$\hat{\gamma}_{\text{wam}}$ vary in these 50 runs, suggesting that the
posterior surface is highly multimodal. The range of the size of
$\hat{\gamma}_{\text{med}}$ and $\hat{\gamma}_{\text{wam}}$ are (2,7)
and (3, 9), respectively. Also, the ranges of $R^2$ values for
$\hat{\gamma}_{\text{med}}$ and $\hat{\gamma}_{\text{wam}}$ are
(27.29\%, 46.85\%) and (32.76\%, 57.18\%), respectively.

When we ran GMC2 for 100 iterations with $\lambda = m/p^2$ and $w =
\sqrt{m}/p$ started at the null model, both
$\hat{\gamma}_{\text{med}}$ and $\hat{\gamma}_{\text{wam}}$ resulted
in the empty model, although the marker with the largest MIP and the
largest weighted MIP was the same SNP $1_{83878775}$ obtained by GMC1
with these choices for $(\lambda, w)$. Indeed, the largest MIP and
weighted MIP were 0.47 and 0.49, respectively. Next, we ran the GMC2
chain for 100 iterations started at the null model with $\lambda =
\sqrt{m}$ and $w = m/p$.  Based on the 100 iterations of the GMC2
chain, $\hat{\gamma}_{\text{med}}$ consisted of eight variables with
$\hat{\gamma}_{\text{med}} = (1_{83878648}, 2_{175541342},
2_{179751269}, 4_{225874150}, 5_{21046482}, 8_{14610995},
8_{72384070},$ $8_{115299982})$ and\\ $\hat{\gamma}_{\text{wam}} =
(1_{83878648}, 2_{175541342}, 3_{170821024}, 4_{237392971},
8_{72384070}, 9_{4338284})$. The $R^2$ values for
$\hat{\gamma}_{\text{med}}$ and $\hat{\gamma}_{\text{wam}}$ in this
case are 43.09\% and 49.37\%, respectively.  It took 11.22 minutes to
complete 100 iterations of GMC2 with $\lambda = \sqrt{m}$ and $w =
m/p$.  Finally, based on 50 repetitions of 100 iterations of the GMC2
chain with $\lambda = \sqrt{m}$ and $w = m/p$, each time with a
different seed, we observe the ranges of the size of
$\hat{\gamma}_{\text{med}}$ and $\hat{\gamma}_{\text{wam}}$ are (4,17)
and (2, 8), respectively. Also, the ranges of $R^2$ values for
$\hat{\gamma}_{\text{med}}$ and $\hat{\gamma}_{\text{wam}}$ are
(40.11\%, 60.10\%) and (28.91\%, 54.61\%), respectively.
\end{example}

\section{Discussion}
\label{sec:disc}
In this work, we have proposed an original framework for developing
informative MCMC schemes. The availability of explicit expressions for
the exponential and inverse exponential maps under the novel use of
square-root representation for pdfs plays a crucial role in
constructing computationally efficient Riemannian manifold geometric
MCMC algorithms. The Riemannian manifold MCMC algorithms available in
the literature work for only continuous targets and involve heavy
computational burden for evaluating the transition densities as well
as for adjusting the tuning parameters. On the other hand, the
proposed geometric MH algorithm works for both discrete and continuous
targets, provides a simple step size tuning as in RWM based on
acceptance rates criteria, and allows a flexible framework for
combining localized steps of the base kernel with cheap to evaluate
local and global approximations of the target to control the
directions of `global' moves of the candidate density while still
using the exact target density in the MH acceptance rate. Indeed, as
demonstrated through examples, the concurrent global-local steps
facilitate moves between the modes and simultaneous exploration of the modal
regions. Thus, the proposed method shows the utility of the explicit
use of the intrinsic geometry of the space of pdfs in constructing
informative MCMC schemes.

Analyses of different high dimensional linear, nonlinear, multimodal
complex statistical models demonstrate the broad applicability of the
proposed method. These examples using both real and simulated data
show that the proposed geometric MCMC algorithms can lead to huge
improvements in mixing and efficiency over alternative MCMC schemes.
The empirical findings corroborate the theoretical results derived
here, comparing the geometric MCMC with the MH chain using the base
proposal density. The theoretical results developed here regarding
different Markov chain orderings hold for general state space Markov
chains and can be used to compare any MCMC algorithms.

The article presents various avenues for potential extensions and
future methodological and theoretical works. The proposed geometric
MCMC scheme is general and can be applied to sample from arbitrary
discrete or continuous targets. The methodology is flexible and allows
general choices for the baseline density and various local/global
approximations of the target for specifying the directions of moving
the baseline proposal density. For the popular Bayesian variable
selection model considered here, we have developed some specific
choices for the densities $f$ and $g$ and are implemented in the
accompanying R package {\it geommc}. In particular, we have used RW
base densities and the target pmf on a neighborhood as $g$. We have
demonstrated the efficiency of the resulting geometric MH algorithm
through analyses of high dimensional simulated datasets and a gigantic
real dataset with close to a million markers for GWAS. One can
consider the Locally Informed and Thresholded proposal of
\cite{zhou:yang:2022} and the adaptively scaled individual adaptation
proposal of \cite{grif:latu:2021} as base densities. In particular, by
allowing multiple variables to be added or deleted from the model in a
single iteration, the algorithm can make large jumps in model space
\citep{guan:step:2011, lian:livi:2022}. For $g$, one can consider
various local tempered and non-tempered versions of $\psi$
\citep{zane:robe:2019}.  We anticipate future works exploring and
comparing different choices of $\mathcal{G}$ mentioned in
Section~\ref{sec:mmh} and developing appropriate $f$ and $ g$
functions for constructing efficient geometric MCMC schemes for other
classes of statistical models. These choices of $f$ and $ g$ densities
can then be implemented in the geommc package for specific
applications.

Algorithm~\ref{alg:mmhmix} uses a fixed $\epsilon$; on the other hand,
one can choose $\epsilon$ adaptively, on the fly. A potential future
work is to build such adaptive versions of
Algorithm~\ref{alg:mmhmix}. Also, it will be interesting to study the
performance of the proposed geometric MH algorithm for the Bayesian
variable selection example in the context of analyzing ordinal responses
\citep{zhen:guo:2023}, most likely extended in a
Metropolis-within-Gibbs framework.  \cite{zhou:yang:2022} prove that the mixing
time of LIT-MH is independent of the number of covariates under the
assumptions of \cite{yang:wain:jord:2016}. \cite{zhou:chan:2023}
derived mixing time bounds for random walk MH algorithms for
high-dimensional statistical models with discrete parameter spaces and
more recently \cite{chan:zhou:2024} extended these results to study
\pcite{zane:2020} informed MH algorithms. A future study is to
undertake such mixing time analysis of our proposed geometric MH
algorithm for Bayesian variable selection.

\smallskip

\noindent{\bf Supplemental Materials}

The supplemental materials contain an alternative formulation of the
geometric MH framework, proofs of theoretical results, additional
technical derivations for some posterior sampling, several plots
corresponding to different examples considered in the main article,
numerical results from an additional data analysis for the Bayesian
logistic regression example, and further simulation studies on the
Bayesian variable selection example.
\bibliographystyle{ims}
\bibliography{C:/Users/vroy/Box/Misclen/LocalTexFiles/bibtex/bib/misc/refs}

\pagebreak
 \centerline{\large\bf A geometric approach to informed MCMC sampling}
 \vspace{.25cm}
 \centerline{Vivekananda Roy} 
\vspace{.4cm} 
\centerline{\it Department of Statistics, Iowa State University, USA}
\vspace{.55cm}
 \centerline{\bf Supplementary Materials}
\vspace{.55cm}

\setcounter{equation}{0}
\setcounter{figure}{0}
\setcounter{table}{0}
\setcounter{page}{1}
\setcounter{section}{0}
\makeatletter
\renewcommand{\thesection}{S\arabic{section}}
\renewcommand{\thesubsection}{\thesection.\arabic{subsection}}
\renewcommand{\theequation}{S\arabic{equation}}
\renewcommand{\thefigure}{S\arabic{figure}}
\renewcommand{\bibnumfmt}[1]{[S#1]}
\renewcommand{\citenumfont}[1]{S#1}
\renewcommand{\thetable}{S\arabic{table}}

\section{An alternative geometric MH algorithm}

We now describe an alternative to Algorithm~\ref{alg:mmhmix}. Denoting this alternative geometric
MH chain by $\{\tilde{X}^{(n)}\}_{n \ge 1}$, suppose $\tilde{X}^{(n-1)} = x$ is the current value
of the chain. The following iterations are used to move to $\tilde{X}^{(n)}$.
\begin{algorithm}[H]
\caption{The $n$th iteration}
\begin{algorithmic}[1]
  \label{alg:mmh}
  \STATE Choose $i' \in \{1,\dots,k\}$ with probability $P(i' =i) =a_i, 1 \le i \le k$.
  \STATE Draw $y \sim \phi_{i',\epsilon} (y|x)$.
  \STATE Set 
    \begin{equation}
      \label{eq:accep2}
    \tilde{\alpha}_{i'}(x, y) = \mbox{min} \Big\{\frac{\psi(y) \phi_{i', \epsilon}(x|y)}{\psi(x) \phi_{i', \epsilon}(y|x)}, 1 \Big\}.
    \end{equation}
  \STATE Draw $\delta \sim$ Uniform $(0, 1)$. If $\delta < \tilde{\alpha}_{i'}(x, y)$
  then set $\tilde{X}^{(n)} \leftarrow y$, else set $\tilde{X}^{(n)} \leftarrow x$.
\end{algorithmic}
\end{algorithm}

\section{Proofs of theoretical results}
\begin{proof}[Proof of Proposition~\ref{prop:perturb}]
We have
  \begin{align}
    \label{eq:sqprop}
    &\exp_{\sqrt{f}}(\epsilon \exp^{-1}_{\sqrt{f}}(\sqrt{g_i}))(y|x)  = \cos\big(\epsilon \theta_{i,x} [\sin(\theta_{i,x})]^{-1}\big\|\sqrt{g_i(y|x)} - \sqrt{f(y|x)} \cos(\theta_{i,x})\big\|\big) \nonumber\\&  \sqrt{f(y|x)} +
                            \sin \big( \epsilon \theta_{i,x} [\sin(\theta_{i,x})]^{-1}\big\|\sqrt{g_i(y|x)} - \sqrt{f(y|x)} \cos(\theta_{i,x})\big\|\big)\epsilon\theta_{i,x} [\sin (\theta_{i,x})]^{-1} \nonumber\\& \big(\sqrt{g_i(y|x)} - \sqrt{f(y|x)} \cos(\theta_{i,x})\big)\Big\{\epsilon\theta_{i,x} [\sin(\theta_{i,x})]^{-1}\big\|\sqrt{g_i(y|x)} - \sqrt{f(y|x)} \cos(\theta_{i,x})\big\|\Big\}^{-1}\nonumber\\
    & = \cos\big(\epsilon \theta_{i,x} [\sin(\theta_{i,x})]^{-1}\big\|\sqrt{g_i(y|x)} - \sqrt{f(y|x)} \cos(\theta_{i,x})\big\|\big) \sqrt{f(y|x)} \nonumber\\& +
  \sin \big( \epsilon \theta_{i,x} [\sin(\theta_{i,x})]^{-1}\big\|\sqrt{g_i(y|x)} - \sqrt{f(y|x)} \cos(\theta_{i,x})\big\|\big) \frac{\big[\sqrt{g_i(y|x)} - \sqrt{f(y|x)} \cos(\theta_{i,x})\big]}{\big\|\sqrt{g_i(y|x)} - \sqrt{f(y|x)} \cos(\theta_{i,x})\big\|}.
  \end{align}
  Note that
  \begin{align}
    \label{eq:norm}
    \big\|\sqrt{g_i(y|x)} - \sqrt{f(y|x)} \cos(\theta_{i,x})\big\|^2 &= \big\|\sqrt{g_i(y|x)}\big\|^2 + \big\langle \sqrt{f(y|x)}, \sqrt{g_i(y|x)} \big\rangle^2 \big\|\sqrt{f(y|x)}\big\|^2\nonumber\\& \quad - 2 \big\langle \sqrt{f(y|x)}, \sqrt{g_i(y|x)} \big\rangle^2 \nonumber\\
    &= 1 - \big\langle \sqrt{f(y|x)}, \sqrt{g_i(y|x)} \big\rangle^2.
  \end{align}
Using the fact that $\sin(\cos^{-1}(r)) = \sqrt{1- r^2},$ we then have
\begin{equation}
    \label{eq:coef}
  \epsilon \theta_{i,x} [\sin(\theta_{i,x})]^{-1}\big\|\sqrt{g_i(y|x)} - \sqrt{f(y|x)} \cos(\theta_{i,x})\big\| = \epsilon \theta_{i,x}.
\end{equation}
Using \eqref{eq:coef} in \eqref{eq:sqprop}, then we have
\begin{equation}
\label{eq:pertrho}
      \exp_{\sqrt{f}}(\epsilon \exp^{-1}_{\sqrt{f}}(\sqrt{g_i}))(y|x)  = \cos(\epsilon \theta_{i,x}) \sqrt{f(y|x)} +
  \sin ( \epsilon \theta_{i,x}) \zeta_i(y|x).
\end{equation}
Next, squaring both sides of \eqref{eq:pertrho} and using the fact that $2 \cos(\epsilon \theta_{i,x}) \sin(\epsilon \theta_{i,x}) = \sin(2\epsilon \theta_{i,x})$ we have \eqref{eq:prop}.
  Now
  \begin{align}
  \label{eq:pdfh}
  \|\zeta_i (y|x)\|^2 &= \langle \zeta_i(y|x), \zeta_i(y|x) \rangle \nonumber\\&= \frac{\big\|\sqrt{g_i(y|x)}\big\|^2 + \big\langle \sqrt{f(y|x)}, \sqrt{g_i(y|x)} \big\rangle^2 \big\|\sqrt{f(y|x)}\big\|^2 - 2 \big\langle \sqrt{f(y|x)}, \sqrt{g_i(y|x)} \big\rangle^2}{1 - \big\langle \sqrt{f(y|x)}, \sqrt{g_i(y|x)} \big\rangle^2} \nonumber\\&=1,
\end{align}
that is, $h_i = \zeta_i^2$ is a pdf.     Note that $0 \le \sin(2\epsilon \theta_{i,x}) \le 1$ as
    $0 \le \theta_{i,x} \le \pi/2$. Since
    \begin{align*}
        \big\langle \sqrt{f(y|x)}, \zeta_i(y|x) \big\rangle &= \frac{\big[\big\langle \sqrt{f(y|x)},
    \sqrt{g_i(y|x)} \big\rangle - \big\|\sqrt{f(y|x)}\big\|^2\big\langle \sqrt{f(y|x)}, \sqrt{g_i(y|x)} \big\rangle\big]}{\sqrt{1
      - \big\langle \sqrt{f(y|x)}, \sqrt{g_i(y|x)} \big\rangle^2}}\\& = \frac{\big[\big\langle \sqrt{f(y|x)},
    \sqrt{g_i(y|x)} \big\rangle - \big\langle \sqrt{f(y|x)}, \sqrt{g_i(y|x)} \big\rangle\big]}{\sqrt{1
      - \big\langle \sqrt{f(y|x)}, \sqrt{g_i(y|x)} \big\rangle^2}} =0,
\end{align*}
      
      $\zeta_i \in T_{\sqrt{f}} \mathcal{M}$. Thus from \eqref{eq:pdfh} it follows that \eqref{eq:prop}
    is a pdf.
\end{proof}
\begin{proof}[Proof of Theorem~\ref{thm:orde}]
 For proving 1 and 2(b), we consider the two cases $c> 1$ and $c \le 1$ separately.

 \noindent{Case $c>1$:}

 \noindent{Proof of 1.} Define the Mtf $P_1 = (1/c)P + (1-1/c) I$. From the assumption
  on $P, Q$, it follows that $P_1 \succeq Q$. Then from \citet[][Lemma
  3]{tier:1998} it follows that
  $\langle P_1 t, t \rangle_\psi \le \langle Qt, t \rangle_\psi $ for all $t \in
  L^2_0(\psi)$. Then the result holds as
  \[
    \langle P_1 t, t \rangle_\psi = \frac{1}{c} \langle P t, t \rangle_\psi +\Big(1-\frac{1}{c}\Big) \sigma^2_t.
    \]

  \noindent{Proof of 2(b).} Since $P$ is reversible with respect to $\psi$, so is
    the Mtf $P_1$. Since $P_1 \succeq Q$ from \citet[][Theorem 4]{tier:1998} it
    follows that $v(t, P_1) \le v(t, Q)$. Now, from \citet[][Corollary
    2.3]{latu:robe:2013} we know that
    $v(t, P_1)= c v(t, P) + (c-1) \sigma^2_t$, or
    \[
      v(t, P) = \frac{1}{c} v(t, P_1) + \frac{1-c}{c} \sigma^2_t.
    \]
    Since $v(t, P_1) \le v(t, Q)$, we have 
    \[
      v(t, P) \le \frac{1}{c} v(t, Q) + \frac{1-c}{c} \sigma^2_t.
    \]

    \noindent{Case $c\le 1$:}

    \noindent{Proof of 1.} Define $Q_1 = cQ + (1-c) I$. From the assumption
  on $P, Q$, it follows that $P \succeq Q_1$. Then from \citet[][Lemma
  3]{tier:1998} it follows that
  $\langle P t, t \rangle_\psi \le \langle Q_1t, t \rangle_\psi$ for all $t \in
  L^2_0(\psi)$. Then the result holds as
  \[
    \langle Q_1 t, t \rangle_\psi = c \langle Q t, t \rangle_\psi +(1-c)E_{\psi} t^2.
  \]

  \noindent{Proof of 2(b).} Since $Q_1$ is reversible with respect to
  $\psi$ and $P \succeq Q_1$ from \citet[][Theorem 4]{tier:1998}, we
  have $v(t, P) \le v(t, Q_1)$. Then the result follows as from
  \citet[][Corollary 2.3]{latu:robe:2013} we know that
    \[
      v(t, Q_1) = \frac{1}{c} v(t, Q) + \frac{1-c}{c} \sigma^2_t.
    \]
\noindent{Proof of 2(a).} Let $L^2_{0,1}(\psi) = \{s \in L^2_0(\psi): E_{\psi}s^2 =1\}$. From \citet[][chap VI]{reth:1993} we have
    \[
      \mbox{Gap}(P) = 1- \sup_{t \in L^2_{0,1}(\psi)} \langle Pt, t \rangle_\psi.
    \]
    From part 1, we know that for $t \in L^2_{0,1}(\psi)$,
    $\langle Pt, t \rangle_\psi \le c \langle Qt, t \rangle_\psi + (1-c)$. Thus, we
    have
    \[
      \mbox{Gap}(P) \ge 1- c\sup_{t \in L^2_{0,1}(\psi)} \langle Qt, t \rangle_\psi -1+c =c \;\mbox{Gap}(Q).
      \]
\end{proof}
\begin{proof}[Proof of Theorem~\ref{thm:peskmmh}]
  Note that
  \begin{align*}
    \psi(x)P_\phi(x, dy) \mu(dx)&= \psi(x)\phi_{\epsilon} (y|x)\alpha(x, y) \mu(dx)\mu(dy)\\
                     &= \min\{\psi(y)\phi_{\epsilon} (x|y), \psi(x)\phi_{\epsilon} (y|x)\} \mu(dx)\mu(dy)\\
                     &= \min\{\psi(y)\sum_{i=1}^k a_i\phi_{i, \epsilon} (x|y), \psi(x)\sum_{i=1}^k a_i \phi_{i,\epsilon} (y|x)\} \mu(dx)\mu(dy)\\
                     &\ge \sum_{i=1}^k a_i \min\{\psi(y)\phi_{i, \epsilon} (x|y), \psi(x)\phi_{i,\epsilon} (y|x)\} \mu(dx)\mu(dy)\\
                     &= \sum_{i=1}^k a_i \min\{\psi(y)[\cos^2(\epsilon \theta_{i,y}) f(x|y) + \sin^2(\epsilon \theta_{i,y}) h_i(x|y)],
    \\ & \hspace{.6in} \psi(x)[\cos^2(\epsilon \theta_{i,x}) f(y|x) + \sin^2(\epsilon \theta_{i,x}) h_i(y|x)]\} \mu(dx)\mu(dy)\\
                     &\ge \sum_{i=1}^k a_ic_{i\epsilon} \min\{\psi(y) f(x|y), \psi(x) f(y|x)\} \mu(dx)\mu(dy)\\
    &=c_\epsilon \psi(x)Q_f(x, dy)\mu(dx),
  \end{align*}
  where the last inequality follows due to $h_i(y|x) = d_i(y|x) f(y|x)$, $h_i(x|y)= d_i(x|y) f(x|y)$ and \eqref{eq:ci}.
\end{proof}
\begin{proof}[Proof of Lemma~\ref{lemm:uenorm}]
Since $\phi(x) = \cos^2(\epsilon \theta) f(x) + \sin^2(\epsilon \theta)
h(x),$ from \eqref{eq:normg} we have
\begin{align*}
  \frac{\phi(x)}{\psi(x)} &= \cos^2(\epsilon \theta) \exp(x-1/2) + \frac{\sin^2(\epsilon \theta)}{1-\exp(-1/4)}\Big(1 - \exp(-1/8) \exp(x/2 - 1/4)\Big)\\
  &=\frac{\sin^2(\epsilon \theta)}{1-\exp(-1/4)} + r(x),
\end{align*}
where
\[
  r(x) = \cos^2(\epsilon \theta) \exp(x-1/2) - [\sin^2(\epsilon
  \theta)\exp(-1/8)/\{1-\exp(-1/4)\}] \exp(x/2 - 1/4).
  \]It can be shown
that $r(x)$ attains minimum at
$c_0 =2 \log (\sin^2(\epsilon \theta) \exp(-1/8) /[2\cos^2(\epsilon \theta)\\
\{1-\exp(-1/4)\}]) + 1/2$. Thus,
\begin{align*}
  \min_x \frac{\phi(x)}{\psi(x)} &= \frac{\sin^2(\epsilon \theta)}{1-\exp(-1/4)} + r(c_0)\\
  &= \frac{\sin^2(\epsilon \theta)}{1-\exp(-1/4)} - \frac{\sin^4(\epsilon \theta) \exp(-1/4)}{4\cos^2(\epsilon \theta)\{1-\exp(-1/4)\}^2},
\end{align*}
which is positive for $\epsilon >0$. Thus, by
Proposition~\ref{prop:uniergo}, the geometric MH algorithm is
uniformly ergodic in this example.  
\end{proof}

\section{Extra plots for the Example 2}

Figure~\ref{fig:cauchyacf1} provides additional autocorrelation
function plots for the Example 2. In particular, it shows that unlike
the RW chains, the geometric MH algorithms with a diffuse normal
density for $g$ lead to fast decaying autocorrelations.

\begin{figure*}[h]
  \begin{minipage}[b]{0.5\linewidth}
    \includegraphics[width=\linewidth]{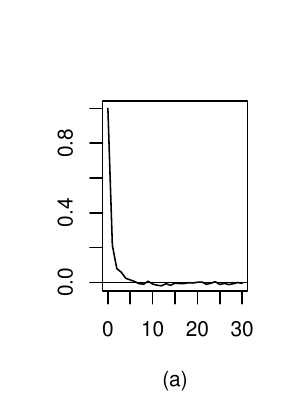}
  \end{minipage}
  \begin{minipage}[b]{0.5\linewidth}
    \includegraphics[width=\linewidth]{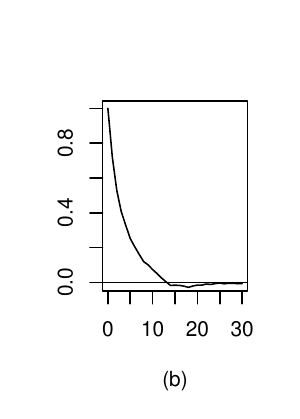}
  \end{minipage}
 \caption{Autocorrelation function plots of the geometric MH chains with $g$ as a diffuse normal density and (a) $t_2$ baseline density and (b) $N(x,1)$ baseline density where $x$ is the current state of the Markov chain. }
\label{fig:cauchyacf1}
\end{figure*}

\section{Extra plots for the Example 4}
Figure~\ref{fig:six_marg.y} shows the marginal histograms obtained
from the samples of the GMC($X_2$)-w-Gibbs chain.
\begin{figure*}[h]
      \begin{minipage}[b]{\linewidth}
    \includegraphics[width=\linewidth]{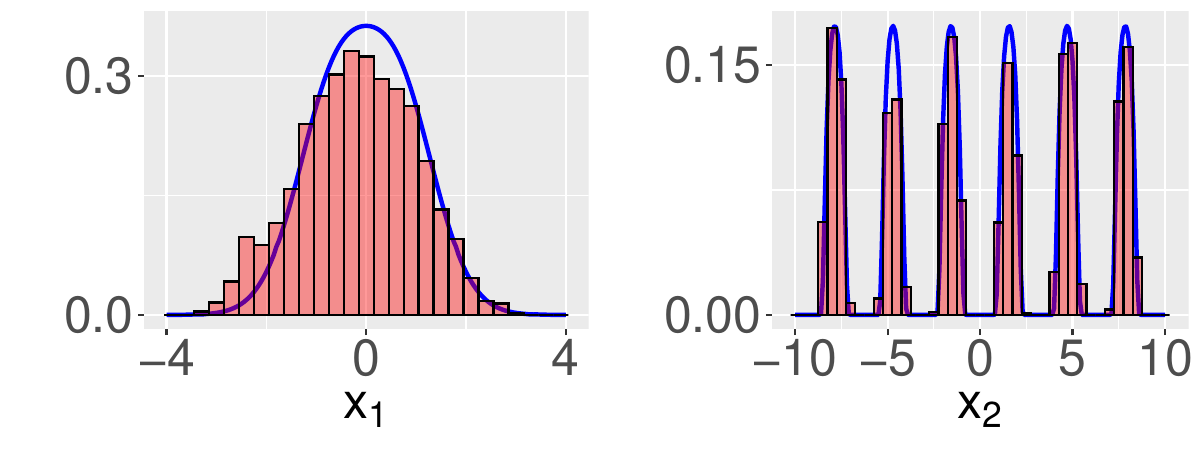}
  \end{minipage}
  \caption{Desired and attained marginals from 100,000 iterations of the GMC($X_2$)-w-Gibbs sampler for the six-mode target example.
    The histograms show that the chain successfully moves between the modes.}
 \label{fig:six_marg.y}
  \end{figure*}

\section{Additional real data example  for the Bayesian logistic regression example}
Here, we consider the German Credit dataset, which contains 
$p = 21$ covariates (including an intercept term), and the response, which is the classification of whether an applicant is considered a good
or bad credit risk based on $m=1000$ loan applicants. We analyze
the German Credit dataset by fitting a Bayesian logistic regression
model using the 10 MCMC algorithms (with appropriate modifications)
considered in Example 5 of the main paper. As in Example 5, we
consider the N$(0_{21}, 10^3 I_{21})$ prior density for $\beta$ and
ran all algorithms for 100,000 iterations, all started at $0_{21}$. For
the RWM, we use a normal proposal with covariance matrix
$\Sigma_f=0.1 \hat{\Sigma}$ to get an acceptance rate of around 50\%.
For MALA and MMALA, we take $h=0.001$ and $h= 1.2$, respectively, resulting in around 50\% acceptance rates. The proposal density
$f_1(\beta)$ is appropriately modified to
$ \phi_{21}(\beta; 0_{21}, \num{e-6}I_{21})$. For the German Credit
dataset, the Ind ($f_1$) chain does not move much, and the R package
mcmcse \citep{R:mcmcse} does not provide the ESS and mESS values.
From Tables~\ref{tab:germaness}-\ref{tab:germanacf}, we see that the
comparative performance of the 10 MCMC samplers remains similar as
observed in Example 5 of the main paper.
  \begin{center}
\begin{table*}[h]
  \caption{Multivariate ESS and ESS (minimum, median, maximum) and their time normalized values for different samplers for the German Credit data}
  \centering
  \begin{tabular}{cccccc}
\hline\hline  Sampler  & mESS  & ESS & mESS/sec  & ESS/sec  &MSJD \\
 \hline
  RW&1282&(785,1136,1622) &21&(13,19,27) &0.008  \\
    GMC(RW)&17563&(7775,17354,19004) &75&(33,74,81) &0.133\\
  Ind&42315&(28485,41480,45140) &515 &(347,505,550)&0.259 \\
  GMC(Ind)& 16652&(9094,15602,18370)&81&(44,76,90)&0.128 \\
  Ind($f_1$)&- &-&-&- &0.000\\
  GMC(Ind ($f_1$))& 15869&(7539,15542, 18804)&75&(36,73,89)&0.129\\
  MALA &3689 &(1236,2648,3880) &20 &(7,14,21)&0.011\\
  GMC(MALA)&18796 &(8981,17700,20389)&33 &(16,31,36)&0.138\\
  MMALA & 22938 &(18817,22470,24204) &4 &(4,4,5)&0.099\\
  GMC(MMALA)&29239 &(21238,28773,31873)&1&(1,1,1)&0.172\\
\hline
\end{tabular}
\label{tab:germaness}
\end{table*}
\end{center}

  \begin{center}
\begin{table*}[h]
  \caption{First eight autocorrelations for different samplers for the German Credit data}
  \centering
\begin{tabular}{ccccccccc}
\hline\hline  Sampler  & lag1  & lag2 & lag3  & lag4 &lag5 &lag6 &lag7 &lag8\\
 \hline
  RW&0.979& 0.958& 0.938& 0.918& 0.898& 0.880& 0.861& 0.843 \\
    GMC(RW)&0.731 & 0.547 &0.421 &0.332 &0.270 &0.224 &0.188& 0.161\\
  Ind&0.463& 0.272& 0.186& 0.134& 0.103& 0.082& 0.067& 0.057\\
  GMC(Ind)&0.742& 0.567& 0.445& 0.359& 0.297& 0.247& 0.207& 0.177\\
  Ind($f_1$)& 0.995& 0.991& 0.988& 0.985& 0.981& 0.978& 0.976& 0.972\\
  GMC(Ind ($f_1$))& 0.742& 0.568& 0.449& 0.364& 0.303& 0.253& 0.214 &0.184\\
  MALA & 0.949& 0.901& 0.856& 0.815& 0.776& 0.739&0.705&0.672\\
  GMC(MALA)&0.722& 0.537& 0.411& 0.326& 0.265& 0.221& 0.187& 0.161\\
  MMALA &0.634& 0.420& 0.287& 0.202& 0.142& 0.098& 0.068& 0.054\\
  GMC(MMALA)&0.594 &0.358 &0.219 &0.135 &0.084 &0.053 &0.033 &0.025\\
\hline
\end{tabular}
\label{tab:germanacf}
\end{table*}
\end{center}  

\section{Detailed steps in MCMC sampling for the Bayesian spatial GLMM example}

The joint posterior density of $(S, \beta,\sigma^2,\theta,\omega)$ is
\begin{align}
  \label{eq:spjoint}
    \psi(s,\beta,\sigma^2,\theta,\omega|z)= \bigg[\prod_{i=1}^m\pi(y_i|s_i)\bigg]\pi(s|\beta,\sigma^2,\theta,\omega)\pi(\beta|\sigma^2)\pi(\sigma^2)\pi(\theta)\pi(\omega),
\end{align}
where
\[
 \pi(y_i|s_i)= \frac{\exp\{-t_ie^{s_i}+s_iz_i\}t_i^{z_i}}{z_i!},
\]
$\pi(s|\beta,\sigma^2,\theta,\omega)$ is the density of
$N_m(X\beta, \Sigma)$, $\pi(\beta|\sigma^2)$ is the density of
$N_p(\mu_0, \sigma^2 \Sigma_0)$, and $\pi(\sigma^2), \pi(\theta), \pi(\omega)$
are the densities of
$\text{Inverse Gamma}(\alpha_1,\gamma_1), \text{Inverse
  Gamma}(\alpha_2,\gamma_2),$ and
$\text{Inverse Gamma}(\alpha_3,\gamma_3)$, respectively.

To construct the various MH-within-Gibbs algorithms discussed in
Example 6 of the main article for \eqref{eq:spjoint}, we now derive
its conditional densities. First, the full conditional density of $S$
is
\begin{equation}
  \label{eq:spscond}
    \psi(s|\beta,\sigma^2,\theta,\omega,z)\propto \bigg[\prod_{i=1}^m \frac{\exp\{-t_ie^{s_i}+s_iz_i\}t_i^{z_i}}{z_i!}\bigg]|\Sigma|^{-1/2}\exp\bigg\{-\frac{1}{2}(s-X\beta)^{\top}\Sigma^{-1}(s-X\beta)\bigg\}.
\end{equation}
For constructing MALA and MMALA for the conditional distribution of $S$, we must derive the first and higher order derivatives of the logarithm of the density \eqref{eq:spscond}. Since
\begin{equation}
  \label{eq:splogscond}
    \log \psi(s|\beta,\sigma^2,\theta,\omega,z)=\sum_{i=1}^m \bigg\{-t_ie^{s_i}+s_iz_i+z_i\log t_i\bigg\}-\frac{1}{2}\log|\Sigma|-\frac{1}{2}(s-X\beta)^{\top}\Sigma^{-1}(s-X\beta)+ c_1,
  \end{equation}
  for some constant $c_1$, we have
\begin{align*}
    \frac{\partial \log \psi(s|\beta,\sigma^2,\theta,\omega,z)}{\partial s}&=-t\cdot e^s+z-\Sigma^{-1}(s-X\beta),\\
    \frac{\partial^2 \log \psi(s|\beta,\sigma^2,\theta,\omega,z)}{\partial s^2}&=-\text{diag}(t\cdot e^s)-\Sigma^{-1}\equiv -G^{-1}, \;\mbox{say},\\
    \mbox{and}\;\frac{\partial \{G^{-1}\}_{(i,j)}}{\partial s_k}&=\begin{cases}
        t_k\exp\{s_k\},&\text{iff} \;i=j=k,\\
        0,&\text{otherwise},
    \end{cases}
\end{align*}
where for two $d$ dimensional vectors $a=(a_1,\dots,a_d)^\top$ and
$b=(b_1,\dots,b_d)^\top$, $a \cdot b =(a_1b_1,\dots,a_d b_d)^\top$, and
$\text{diag}(a)$ denotes the $d\times d$ diagonal matrix with diagonal
elements $a$.

From \eqref{eq:spjoint} we see that conditional on $s$, the
distributions of $\beta, \sigma^2, \theta$ and $\omega$ do not depend
on $z$. The conditional density of $\beta$ is
\begin{align*}
     \psi(\beta|s,\sigma^2,\theta,\omega) \propto&\pi(s|\beta,\sigma^2,\theta,\omega)\pi(\beta|\sigma^2)\\
     \propto&
      \exp\bigg\{-\frac{1}{2\sigma^2}(s-X\beta)^{\top}\Sigma_{\theta,\omega}^{-1}(s-X\beta)\bigg\}\exp\bigg\{-\frac{1}{2\sigma^2}(\beta-\mu_{0})^{\top}\Sigma_{0}^{-1}(\beta-\mu_{0})\bigg\},
\end{align*}
where $\Sigma=\sigma^2\Sigma_{\theta,\omega}$.
Thus,
\begin{equation}
  \label{eq:spbetacond}
      \beta|s,\sigma^2,\theta,\omega \sim N_p\bigg\{\bigg(X^{\top}\Sigma_{\theta,\omega}^{-1}X+\Sigma_{0}^{-1}\bigg)^{-1}\bigg(X^{\top}\Sigma_{\theta,\omega}^{-1}s+\Sigma_{0}^{-1}\mu_{0}\bigg),\sigma^2\bigg(X^{\top}\Sigma_{\theta,\omega}^{-1}X+\Sigma_{0}^{-1}\bigg)^{-1}\bigg\}.
\end{equation}
Next,
\begin{align*}
    \psi(\sigma^2|s,\beta,\theta,\omega)\propto& \pi(s|\beta,\sigma^2,\theta,\omega)\pi(\beta|\sigma^2)\pi(\sigma^2)\\
    \propto& (\sigma^2)^{-m/2}(\sigma^2)^{-\alpha_1-1}(\sigma^2)^{-p/2}\exp\bigg\{-\frac{\gamma_1}{\sigma^2}\bigg\}\exp\bigg\{-\frac{1}{2\sigma^2}(s-X\beta)^{\top}\Sigma_{\theta,\omega}^{-1}(s-X\beta)\bigg\}\\
    &\exp\bigg\{-\frac{1}{2\sigma^2}(\beta-\mu_0)^{\top}\Sigma_{0}^{-1}(\beta-\mu_{0})\bigg\},
\end{align*}
 that is,
 \begin{equation}
   \label{eq:spsigcond}
    \sigma^2|s,\beta, \theta, \omega \sim \text{Inverse Gamma}\bigg(\frac{2\alpha_1+m+p}{2},\frac{2\gamma_1+(s-X\beta)^{\top}\Sigma_{\theta,\omega}^{-1}(s-X\beta)+(\beta-\mu_{0})\Sigma_{0}^{-1}(\beta-\mu_{0})}{2}\bigg).
\end{equation}
The conditional pdf of $\theta$ is
\begin{align*}
     \psi(\theta|s,\beta,\sigma^2,\omega)\propto& \pi(s|\beta,\sigma^2,\theta,\omega)\pi(\theta)\\
     \propto&|\Sigma|^{-1/2}\exp\bigg\{-\frac{1}{2}(s-X\beta)^{\top}\Sigma^{-1}(s-X\beta)\bigg\}\theta^{-\alpha_2-1}\exp\bigg\{-\frac{\gamma_2}{\theta}\bigg\},
\end{align*}
which is not a standard density. Letting $\nu=\log\theta$, we use a RWM step for sampling from $\psi(\nu|s,\beta,\sigma^2,\omega)$
which is given by
\begin{equation}
  \label{eq:spnucond}
      \psi(\nu|s,\beta,\sigma^2,\omega)\propto |\Sigma|^{-1/2}\exp\bigg\{-\frac{1}{2}(s-X\beta)^{\top}\Sigma^{-1}(s-X\beta)\bigg\}\exp\{-\alpha_2\nu\}\exp\bigg\{-\frac{\gamma_2}{\exp\{\nu\}}\bigg\}.
\end{equation}
Finally, the conditional pdf of $\omega$ is 
\begin{align*}
    \psi(\omega|s,\beta,\sigma^2,\theta)&\propto \pi(s|\beta,\sigma^2,\theta,\omega)\pi(\omega)\\
    &\propto|\Sigma|^{-1/2}\exp\bigg\{-\frac{1}{2}(s-X\beta)^{\top}\Sigma^{-1}(s-X\beta)\bigg\}\omega^{-\alpha_3-1}\exp\bigg\{-\frac{\gamma_3}{\omega}\bigg\}.
\end{align*}
Letting $\zeta=\log\omega$, we use a RWM step for sampling from $\psi(\zeta|s,\beta,\sigma^2,\theta)$ given by
\begin{equation}
  \label{eq:spzetcond}
   \psi(\zeta|s,\beta,\sigma^2,\theta)\propto |\Sigma|^{-1/2}\exp\bigg\{-\frac{1}{2}(s-X\beta)^{\top}\Sigma^{-1}(s-X\beta)\bigg\}\exp\{-\alpha_3\zeta\}\exp\bigg\{-\frac{\gamma_3}{\exp\{\zeta\}}\bigg\}.
\end{equation}
Starting with some initial values, the different variables are
sequentially updated in the following order using MH-within-Gibbs
algorithms.

\begin{itemize}
\item[Step 1] Update $s$ with one of the MH algorithms mentioned in Example 6 in the main article.
\item[Step 2] Update $\beta$ with a draw from \eqref{eq:spbetacond}.
 \item[Step 3] Update $\sigma^2$ with a draw from \eqref{eq:spsigcond}. 
\item[Step 4] Update $\theta (\nu)$ by drawing
  \begin{itemize}
  \item[(a)] a sample $\nu' \sim N(\nu, h_\nu)$, and
\item[(b)] accepting $\nu'$ with probability
\[
    \text{min}\bigg\{\frac{\psi(\nu'|s,\beta,\sigma^2,\omega)}{\psi(\nu|s,\beta,\sigma^2,\omega)},1\bigg\},
\]
 where $\psi(\nu|s,\beta,\sigma^2,\omega)$ is given in \eqref{eq:spnucond}.
  \end{itemize}
\item[Step 5] Update $\omega (\zeta)$ by drawing
  \begin{itemize}
  \item[(a)] a sample $\zeta' \sim N(\zeta, h_\zeta)$, and
    \item[(b)] accepting $\zeta'$ with probability
\[
    \text{min}\bigg\{\frac{\psi(\zeta'|s,\beta,\sigma^2,\theta)}{\psi(\zeta|s,\beta,\sigma^2,\theta)},1\bigg\},
\]
 where $\psi(\zeta|s,\beta,\sigma^2,\theta)$ is given in \eqref{eq:spzetcond}.
\end{itemize}
\end{itemize}
The step-sizes $h_\nu$ and $h_\zeta$ are chosen to achieve certain empirical acceptance rates.

\section{Further simulation results for the Bayesian variable
  selection example}
This section presents the results for different simulation settings corresponding to $R^2 = 0.75$ and $R^2 = 0.6$.
\begin{table}[h]
\begin{center}
\small\addtolength{\tabcolsep}{-2pt}
\caption{Results for different samplers for the Bayesian variable selection example ($R^2 = 75\%$).}
\label{tab:rsqpt75}
\begin{tabular}{lrrrrrrrrrr}
  \hline\hline
        &Success&$N_{\text{success}}$&Time 
       &MSPE     & MSE$_\beta$
        &\shortstack{Model \\ size} &Coverage
       & FDR
       & FNR
        &\shortstack{Jaccard \\ Index} \\
   \hline
     & \multicolumn{10}{c}{Independent design}\\
  \hline
  RW1  &   55& 32047& 126.19& 3.067& 1.193 &3.33 & 16 & 0.0 & 33.4 &66.6\\
  GMC1 & 100&     9 &  0.92 & 1.921& 0.038& 4.94 &94&    0.0& 1.2 & 98.8\\
  RW2 &  41& 31811& 119.83 & 3.432 & 1.525 &3.09 &11  &  0.0 & 3.8& 61.8\\
 GMC2 & 100 & 10& 0.97 &1.921 &0.038 &4.94 &94 &0.0 &1.2 &98.8\\
	\hline
 & \multicolumn{10}{c}{Compound symmetry design with $ r= 0.6$}\\
  \hline
  RW1  &   48 & 35349 & 179.56 & 182.214 & 89.068 & 3.14 &0 &7.1 & 41.6 & 57.0\\
  GMC1 & 100 &    8&  0.72 & 163.623 &50.076 &3.68 &3 &1.8 & 27.8 &71.5\\
  RW2 & 62 & 30539 & 146.02 & 186.615 & 97.337 & 2.94 & 0 & 6.1 & 44.6 &54.3\\
  GMC2 & 99 & 8 & 0.69 &164.830 &53.720 &3.69 &2 &4.0 &29.0 &70.0\\
	\hline
         & \multicolumn{10}{c}{Autoregressive correlation design with $ r= 0.6$}\\
  \hline
  RW1  &   69 &31807& 129.17 & 7.957 & 2.060 & 2.64& 56 &4.3 &16.3 &81.6\\
  GMC1 &  100 &  6 & 0.48 & 6.443 & 0.063 & 3.01 &100 & 0.2 &0.0 &99.8\\
   RW2 &  80 & 23211 &  86.97 & 8.948 & 3.823 & 2.62 & 48 & 10.8 &21.7 &74.8\\
  GMC2 &100& 6&  0.45 &6.443 &0.063 &3.01 &100 &0.2 &0.0 &99.8\\
        \hline
         & \multicolumn{10}{c}{Factor model design}\\
  \hline
  RW1  & 48 & 33337 & 157.63 & 194.104 & 49.328 & 3.47 &23 &3.6 &32.4 &67.0\\
  GMC1 &74 &11&   1.10 &160.975 &31.622 &4.01 &64 &4.7 &21.4 &78.3\\
  RW2 &56 &36078 & 154.99 & 206.560 & 60.587 &3.06 &13 &1.7 &39.8 &59.9\\
  GMC2 & 97 &10 & 0.97 &133.696 & 8.144& 4.75 &87 &0.0 &5.0 &95.0\\
        \hline
         & \multicolumn{10}{c}{Extreme correlation design}\\
  \hline
  RW1  & 36 &39268 &176.32 &84.674 &51.900 &3.28 &17 &7.1 &37.4 &61.8\\
  GMC1 & 56 & 25  & 2.04 &84.255 &79.184 &2.97 &46 &35.8 &49.4 &50.0\\
 RW2 & 47 &31541 &134.22& 90.094 &53.652 &3.08 &15 &4.4 &40.2 &59.3\\
GMC2 &100& 12& 1.16& 42.547 &0.535 &5.00 &100& 0.0& 0.0& 100.0\\
        \hline
\hline
\end{tabular}
\end{center}
\end{table}

\begin{table}
\begin{center}
\small\addtolength{\tabcolsep}{-2pt}
\caption{Results for different samplers for the Bayesian variable selection example ($R^2 = 60\%$).}
\label{tab:rsqpt6}
\begin{tabular}{lrrrrrrrrrr}
  \hline\hline
        &Success&$N_{\text{success}}$&Time 
       &MSPE     & MSE$_\beta$
        &\shortstack{Model \\ size} &Coverage
       & FDR
       & FNR
        &\shortstack{Jaccard \\ Index} \\
   \hline
     & \multicolumn{10}{c}{Independent design}\\
  \hline
  RW1  & 65 & 32864 &123.50 & 5.119 & 1.353 & 3.00 & 5& 0.0 & 40.0 &60.0\\
  GMC1 &  99 & 8  & 0.74 & 3.987& 0.233 & 4.28 &31& 0.0 &14.4 & 85.6\\
  RW2 & 47 & 27859 & 100.69 & 5.523 & 1.748 &2.75 &1  &  0.0 &45.0 &55.0\\
 GMC2 & 100& 8& 0.71 &3.974 &0.217 &4.32 &32 & 0.0 &13.6 &86.4\\
	\hline
 & \multicolumn{10}{c}{Compound symmetry design with $ r= 0.6$}\\
  \hline
  RW1  &  24 & 28640 & 124.54 & 416.001 & 220.563 & 1.76 & 0 &29.3 &76.0 & 22.7\\
  GMC1 & 94 & 11& 0.86 & 374.674 & 195.797 & 2.15 & 0 & 25.8 & 68.6 &29.6\\
  RW2 & 55 & 26374 & 109.37 & 395.957 &217.847 &1.84 &  0 &27.3 & 74.0 &24.5\\
  GMC2 & 89 & 9 &0.67 & 373.692 & 198.911 &2.19 &0 &28.2 & 68.8 &29.1\\
	\hline
         & \multicolumn{10}{c}{Autoregressive correlation design with $ r= 0.6$}\\
  \hline
  RW1  &   66 & 28428 & 108.76 & 14.870 &2.480 &2.47 & 47 & 3.6& 21.0 &77.4\\
  GMC1 &  100&  6& 0.47 &12.936 &0.174 &2.99 &98 &0.2 &0.7 &99.1\\
  RW2 & 70 &23635 &85.47 &15.817 &4.096 &2.41 &42 &8.5 &25.7 &72.6\\
  GMC2 &100& 6&0.45 &12.914 &0.151 &3.00 &99 &0.2 &0.3 &99.4\\
        \hline
         & \multicolumn{10}{c}{Factor model design}\\
  \hline
  RW1  &  45 & 35221 & 142.32& 348.548 & 74.650 &2.30 &6 &8.7 &50.4 &42.0\\
  GMC1 & 60  & 10 & 0.65 & 330.362 &78.499 &2.25 &21 &18.0 &52.0 &40.9\\
  RW2 & 53 & 27323 & 101.88& 369.089 &81.819 &1.94 &7 &7.1 &56.2 &36.5\\
  GMC2 & 89 & 10  &0.85& 290.329 &50.733 &3.25 &46 &11.5 &31.2 &61.5\\
        \hline
         & \multicolumn{10}{c}{Extreme correlation design}\\
  \hline
  RW1  &47 &35036 & 172.97 &39.850 &27.131 & 4.02 &36 &1.4 &20.8 &78.5\\
  GMC1 & 100& 12& 1.23& 16.968 &6.563 &4.85 &96 &3.5 &3.8&96.2\\
  RW2 & 58&31982 &147.35& 48.720 &37.798 &3.70 &19 &3.1 &28.4&70.3\\
  GMC2 &100 &11 &1.07 & 14.182& 0.178&5.00&100 &0.0 &0.0 &100.0\\
        \hline
\hline
\end{tabular}
\end{center}
\end{table}


\end{document}